\documentclass{eptcs}
\usepackage{breakurl}
\usepackage{latexsym}                   
\usepackage{enumerate}
\providecommand{\publicationstatus}{An extended abstract of this paper
  appeared in the proceedings of FoSSaCS 2018 \cite{vG18b}.}
\makeatletter
\newif\if@qeded
\global\@qededtrue
\def\qed{\hfill$\Box$\global\@qededtrue}
\def\qedneeded{\global\@qededfalse}
\def\qedifneeded{\if@qeded\else\qed\fi}
\makeatother
\newtheorem{defi}{Definition}
\newtheorem{theo}{Theorem}
\newtheorem{prop}{Proposition}
\newtheorem{lemm}{Lemma}
\newtheorem{coro}{Corollary}
\newtheorem{exam}{Example}
\newtheorem{obse}{Observation}
\newtheorem{post}{Postulate}

\newenvironment{definition}[1]{\begin{defi} \rm \label{df:#1} }{\end{defi}}
\newenvironment{definitionc}[2]{\begin{defi}[#2] \rm \label{df:#1} }{\end{defi}}
\newenvironment{theorem}[1]{\begin{theo} \rm \label{thm:#1} }{\end{theo}}

\newenvironment{proposition}[1]{\begin{prop} \rm \label{pr:#1} }{\end{prop}}
\newenvironment{propositionc}[2]{\begin{prop}[#2] \rm \label{pr:#1} }{\end{prop}}
\newenvironment{lemma}[1]{\begin{lemm} \rm \label{lem:#1} }{\end{lemm}}
\newenvironment{corollary}[1]{\begin{coro} \rm \label{cor:#1} }{\end{coro}}
\newenvironment{example}[1]{\begin{exam} \rm \label{ex:#1} }{\end{exam}}
\newenvironment{observation}[1]{\begin{obse} \rm \label{obs:#1} }{\end{obse}}
\newenvironment{postulate}[1]{\begin{post} \rm \label{post:#1} }{\end{post}}
\newenvironment{postulatec}[2]{\begin{post}[#2] \rm \label{post:#1} }{\end{post}}
\newenvironment{proof}{\qedneeded\begin{trivlist} \item[\hspace{\labelsep}\bf Proof:]}
               {\qedifneeded\end{trivlist}}

\newcommand{\df}[1]{Definition~\ref{df:#1}}
\newcommand{\thm}[1]{Theorem~\ref{thm:#1}}
\newcommand{\pr}[1]{Proposition~\ref{pr:#1}}
\newcommand{\lem}[1]{Lemma~\ref{lem:#1}}
\newcommand{\cor}[1]{Corollary~\ref{cor:#1}}
\newcommand{\ex}[1]{Example~\ref{ex:#1}}
\newcommand{\obs}[1]{Observation~\ref{obs:#1}}
\newcommand{\sect}[1]{Section~\ref{sec:#1}}

\newcommand{\pos}[1]{Postulate~\ref{post:#1}}

\makeatletter
\def\@listi{\leftmargin\leftmargini
            \parsep 0\p@ \@plus1.5\p@ \@minus0\p@
            \topsep 2\p@   \@plus2\p@ \@minus2\p@
            \itemsep0\p@ \@plus1.5\p@ \@minus0\p@}
\let\@listI\@listi
\@listi
\makeatother

\DeclareSymbolFont{cmi}   {OT1}{cmr}{m}{it}
\DeclareMathSymbol\val    {\mathord}{cmi}{118} 
\DeclareMathSymbol\wal    {\mathord}{cmi}{119} 
\newfont{\bbb}{bbm10 scaled 1100}        
\newfont{\bbbs}{bbm10 scaled 900}        
\newcommand{\denote}[1]{\mbox{\bbb [}#1\mbox{\bbb ]}} 
\newcommand{\IO}{\mbox{\bbb O}}          
\newcommand{\IOs}{\mbox{\bbbs O}}        
\newcommand{\IN}{\mbox{\bbb N}}          
\newcommand{\INs}{\mbox{\bbbs N}}        
\newcommand{\IT}{\mbox{\bbb T}}          
\newcommand{\T}{{\rm T}}                 
\newcommand{\N} {{\cal N}}               
\newcommand{\D}{{\cal D}}                
\newcommand{\fL}{{\cal L}}               
\newcommand{\fT}{{\cal T}}               
\newcommand{\bR}{\mathrel{\bf R}}        
\newcommand{\bT}{{\bf T}}                
\newcommand{\bU}{{\bf U}}                
\newcommand{\bV}{{\bf V}}                
\newcommand{\bW}{{\bf W}}                
\newcommand{\bZ}{{\bf Z}}                
\newcommand{\E}{E}                       
\newcommand{\p}{P}                       
\newcommand{\V}{\mathcal{X}}             
\newcommand{\n}{{\sf n}}                 
\newcommand{\bn}{{\sf bn}}               
\newcommand{\fv}{{\it fv}}               
\newcommand{\plat}[1]{\raisebox{0pt}[0pt][0pt]{#1}}     
\newcommand{\dom}{{\it dom}}                            
\newcommand{\range}{{\it range}}                        
\newcommand{\eqa}{\mathrel{\plat{$\stackrel{\alpha}=$}}} 
\newcommand{\asim}{\mathrel{\mbox{${\scriptscriptstyle\bullet}\hspace{-1.15ex}\sim$}}}
\newcommand{\subs}[2]{\{\mathord{\raisebox{2pt}[0pt]{$#1$}\!/\!#2}\}} 
\newcommand{\sbarb}[2]{#1{\downarrow_{#2}}}
\newcommand{\wbarb}[2]{#1{\Downarrow_{#2}}}
\newcommand{\nwbarb}[2]{#1{\not\Downarrow_{#2}}}
\newcommand{\wbb}{\stackrel{\raisebox{-1pt}[0pt][0pt]{$\scriptscriptstyle \bullet$}}{\approx}}
\newcommand{\sbb}{\stackrel{\raisebox{-1pt}[0pt][0pt]{$\scriptscriptstyle \bullet$}}{\sim}}
\newcommand{\bis}[1]{\mathrel{\,		      
	\raisebox{.3ex}{$\underline{\makebox[.7em]{$\leftrightarrow$}}$}
                  \,_{#1}}}
\newcommand{\bbbis}{\stackrel{\makebox[0pt]{\raisebox{-1.5pt}[0pt][0pt]{$\scriptscriptstyle \bullet\hspace{.45pt}$}}}{\bis{}}^\Delta_b}
\def\titlerunning{A Theory of Encodings and Expressiveness}
\title{\titlerunning}
\author{Rob van Glabbeek
\institute{Data61, CSIRO, Sydney, Australia}
\institute{School of Computer Science and Engineering,
University of New South Wales, Sydney, Australia}
\email{rvg@cs.stanford.edu}
}
\def\authorrunning{R.J. van Glabbeek}
\begin{document}
\maketitle

\begin{abstract}
This paper proposes a definition of what it means for one system
description language to encode another one, thereby enabling an ordering
of system description languages with respect to expressive power.
I compare the proposed definition with other definitions of encoding and
expressiveness found in the literature, and illustrate it on a well-known
case study: the encoding of the synchronous in the asynchronous $\pi$-calculus.
\end{abstract}

\section{Introduction}

This paper, like \cite{Gorla10a,vG12}, aims at answering the question what it means for one language to
encode another one, and making the resulting definition applicable to order system description
languages like CCS, CSP and the $\pi$-calculus with respect to their expressive power.

To this end it proposes a unifying concept of valid translation between two languages
\emph{up to} a semantic equivalence or preorder.
It applies to languages whose semantics interprets the operators and recursion
constructs as operations on a set of values, called a \emph{domain}.
Languages can be partially ordered by their expressiveness up to the chosen equivalence or preorder
according to the existence of valid translations between them.

The concept of a [valid] translation between system description languages (or \emph{process
calculi}) was first formally defined by Boudol \cite{Bo85}. There, and in most other
related work in this area, the domain in which a system description language is
interpreted consists of the closed expressions from the language itself. In \cite{vG94a} I
have reformulated Boudol's definition, while dropping the requirement that the domain of
interpretation is the set of closed terms. This allows (but does not
enforce) a clear separation of syntax and semantics, in the tradition of universal
algebra.  Nevertheless, the definition employed in \cite{vG94a} only deals with the case
that all (relevant) elements in the domain are denotable as the interpretations of closed
terms. In \cite{vG12} situations are described where such a restriction is undesirable.
In addition, both \cite{Bo85} and \cite{vG94a} require the semantic equivalence $\sim$ under
which two languages are compared to be a congruence for both of them.
This is too severe a restriction to capture many recent encodings
\cite{
Boreale98,        
Nestmann00,       
NestmannP00,      
CarboneM03,       
BPV04,            
BPV05,            
PalamidessiSVV06, 
PN12}.            

In \cite{vG12} I alleviated these two restrictions by proposing two notions of encoding:
\emph{correct} and \emph{valid} translations up to $\sim$. Each of them generalises the
proposals of \cite{Bo85} and \cite{vG94a}. The former drops the
restriction on denotability as well as $\sim$ being a congruence for the whole target language,
but it requires $\sim$ to be a congruence for the source language, as well as for the source's
image within the target. The latter drops both congruence requirements (and allows $\sim$ to be a
preorder rather than an equivalence), but at the expense of requiring denotability by closed terms.
In situations where $\sim$ is a congruence for the source language's image within the target
language \emph{and} all semantic values are denotable, the two notions agree.

The current paper further generalises the work of \cite{vG12} by proposing a new notion of a valid
translation that incorporates the correct and valid translations of \cite{vG12} as special cases.
It drops the congruence requirements as well as the restriction on denotability.

As in \cite{vG12}, my aim is to generalise the concept of a valid translation as much as possible,
so that it is uniformly applicable in many situations, and not just in the world of process
calculi. Also, it needs to be equally applicable to encodability and separation results,
the latter saying that an encoding of one language in another does not exists.
At the same time, I try to derive this concept from a unifying principle,
rather than collecting a set of criteria that justify a number of known
encodability and separation results that are intuitively justified.

\paragraph{Overview of the paper}
\sect{validity} defines my new concept of a valid translation up to a semantic equivalence or
preorder $\asim$.
Roughly, a valid translation up to $\asim$ of one language into another is
a mapping from the expressions in the first language to those in the second
that preserves their meaning, i.e.\ such that the meaning of the translation of
an expression is semantically equivalent to the meaning of the expression being translated. 

\sect{valid-correct} shows that this concept generalises the notion of a correct translation from \cite{vG12}:
a translation is correct up to a semantic equivalence $\sim$ iff it is valid up to $\sim$ and $\sim$
is a congruence for the source language as well as for the image of the source language within the
target language.

Likewise, \sect{respects} shows that my new concept of validity generalises the one of \cite{vG12},
and \sect{fvr} establishes the coincidence of my new validity-based notion of expressiveness with
the one from \cite{vG12} when applying both to languages for which all semantic values are denotable
by closed terms.

One language is said to be at least as expressive as another up to $\asim$
iff there exists a valid translation up to $\asim$ of the latter language into the former.
\sect{hierarchy} shows that the relation ``being at least as expressive as'' is a
preorder on languages.  This expressiveness preorder depends on the choice of $\asim$, and
a coarser choice (making less semantic distinctions) yields a richer preorder of expressiveness
inclusions.

\sect{closed-term} presents the widely used class of \emph{closed-term} languages,
in which the distinction between syntax and semantic is effectively dropped by taking the domain of
values where the language is interpreted in to consist of the closed terms of the language.
\sect{asynpi} illustrates my approach on a well-known case study:
the encoding of the synchronous in the asynchronous $\pi$-calculus.

\sect{congruence closure} discusses the \emph{congruence closure} of a semantic
equivalence for a given language, and remarks that in the presence of operators with infinite arity
it is not always a congruence.
\sect{ccproperty} states a useful congruence closure property for valid translations:
if a translation between two languages exists that is valid up a semantic equivalence $\sim$,
then it is even valid up to an equivalence that
{\leftmargini 12pt
\begin{itemize}
\item on the source language coincides with the congruence closure of $\sim$
\item on the image of the source within the target language also coincides with the congruence closure of $\sim$%
\item melts each equivalence class of the source with exactly one of the target, and vice versa.
\end{itemize}}

\sect{integrating} concludes that the framework established thus far is very suitable for comparing the
expressiveness of languages, but falls short for the purpose of combining language features.
This requires a congruence reflection theorem, provided in \sect{reflect}, for closed-term languages
 and preorders $\asim$ that satisfy some mild sanity requirements:
the postulates formulated in Sections~\ref{sec:standard heads} and~\ref{sec:invariance}.

\sect{compositionality} defines when a translation is \emph{compositional}, and shows that any
valid translation up to $\asim$ can be modified into a compositional translation valid up to
$\asim$, provided the languages and preorders $\asim$ satisfy the sanity requirements of
Sections~\ref{sec:standard heads} and~\ref{sec:invariance}.
Hence, for the purpose of comparing the expressive power of languages,
valid translations between them may be presumed compositional.

\sect{preservation} contemplates a more general, and arguably also simpler, concept of a valid translation,
than the one of \sect{validity}. However, it lacks appealing properties of the latter.

\sect{barbed} speculates on suitable choices for $\asim$ when comparing the expressiveness of
process calculi.
Sections~\ref{sec:full abstraction}--\ref{sec:related} compare my approach with
\emph{full abstraction}, and with the approach of Gorla~\cite{Gorla10a}.

\section{Languages, correct and valid translations, and expressiveness}\label{sec:validity}

A language consists of \emph{syntax} and \emph{semantics}.
The syntax determines the valid expressions in the language.
The semantics is given by a mapping $\denote{\ \ }$ that associates
with each valid expression its meaning, which can for instance be an
object, concept or statement.

Following \cite{vG12}, I represent a language $\fL$ as a pair
$(\IT_{\fL},\denote{\ \ }_\fL)$ of a set $\IT_{\fL}$ of valid expressions in
$\fL$ and a mapping $\denote{\ \ }_\fL:\IT_\fL\rightarrow \D_\fL$
from $\IT_\fL$ in some set of meanings $\D_\fL$.

\begin{definitionc}{translation}{\cite{vG12}}
A \emph{translation} from a language $\fL$ into a language $\fL'$ is a
mapping $\fT: \IT_\fL \rightarrow \IT_{\fL'}$. 
\end{definitionc}
In this paper, I consider single-sorted languages $\fL$ in which
\emph{expressions} or \emph{terms} are built from variables (taken
from a set $\V$) by means of operators (including constants) and
possibly recursion constructs. For such languages the meaning $\denote{\E}_\fL$
of an $\fL$-expression $\E$ is a function of type
$(\V\!\!\rightarrow\bV)\rightarrow\bV$ for a given sets of \emph{values} $\bV$.
It associates a value $\denote{E}_\fL(\rho) \mathbin\in\bV$ to $E$ that depends on
the choice of a \emph{valuation} $\rho\!:\V\!\!\!\rightarrow\!\bV$. The
valuation associates a value from $\bV$ with each variable.

Since normally the names of variables are irrelevant and the
cardinality of the set of variables satisfies only the requirement
that it is ``sufficiently large'', no generality is lost by insisting
that two (system description) languages whose expressiveness is being
compared employ the same set of (process) variables.\linebreak[2]
On the other hand, two languages $\fL$ and $\fL'$ may be interpreted in
different domains of values $\bV$ and $\bV'\!$.

Let $\fL$ and $\fL'$ be two languages of the type considered above,
with semantic mappings\\[1ex]
\mbox{}\hfill
$\denote{\ \ }_{\fL}:\IT_{\fL} \rightarrow ((\V\rightarrow\bV)\rightarrow\bV)$
\hfill
and
\hfill
$\denote{\ \ }_{\fL'}:\IT_{\fL'} \rightarrow ((\V\rightarrow\bV')\rightarrow\bV')$.
\hfill\mbox{}\\[1ex]
In order to compare these languages
w.r.t.\ their expressive power I need a semantic equivalence or
preorder $\asim$ that is defined on a unifying domain of interpretation $\bZ$,
with $\bV,\bV' \subseteq\bZ$.\footnote{I will be chiefly interested in the case that $\asim$ is an
  equivalence---hence the choice of a symbol that looks like $\sim$. However, to establish
  \obs{identity} and \thm{composition} below, it suffices to know that $\asim$ is reflexive and
  transitive. My convention is that the dotted end of \mbox{$\asim$} points to a translation
  and the other end to an original---without offering an intuition for the possible asymmetry.}
Intuitively, $\val'\asim \val$ with $\val\in\bV$ and $\val'\in\bV'$ means that values $\val$ and
$\val'$ are sufficiently alike for our purposes, so that one can accept a translation of an
expression with meaning $\val$ into an expression with meaning $\val'$.

\emph{Correct} and a \emph{valid} translations up to a semantic equivalence or preorder
$\asim$ were introduced in \cite{vG12}. Here I redefine these concepts in terms of a new concept of
\emph{correctness w.r.t.\ a semantic translation}.

\begin{definition}{semantic translation}
Let $\bV$ and $\bV'$ be domains of values in which two languages $\fL$ and $\fL'$ are interpreted.\\
A \emph{semantic translation} from $\bV$ into $\bV'$ is a relation $\mathord{\bR}\subseteq\bV'\times\bV$
such that $\forall \val\in\bV.\,\exists \val'\in\bV'.\,\val'\bR \val$.
\end{definition}
Thus every semantic value in $\bV$ needs to have a counterpart in $\bV'$---possibly multiple ones.\\
For valuations $\eta:\V\rightarrow\bV'$, $\rho:\V\rightarrow\bV$ I write $\eta\bR\rho$ iff
$\eta(X)\bR\rho(X)$ for each $X\in \V$.
\begin{definition}{correct R}
A translation $\fT:\IT_\fL\rightarrow\IT_{\fL'}$ is \emph{correct} w.r.t.\ a semantic translation $\bR$
if $\denote{\fT(\E)}_{\fL'}(\eta) \bR \denote{\E}_\fL(\rho)$ for all expressions $\E\in \IT_\fL$
and all valuations $\eta:\V\rightarrow\bV'$ and $\rho:\V\rightarrow\bV$ with $\eta\bR\rho$.
\end{definition}
Thus $\fT$ is correct iff the meaning of the translation of an expression $E$ is a counterpart
of the meaning of $E$, no matter what values are filled in for the variables,
provided that the value filled in for a given variable $X$ occurring in the translation $\fT(E)$ is
a counterpart of the value filled in for $X$ in $E$.
\begin{definition}{correct}
A translation $\fT:\IT_\fL\rightarrow\IT_{\fL'}$ is \emph{correct} up to $\asim$ iff
$\asim$ is an equivalence, the restriction $\bR$ of $\asim$ to $\bV'\times\bV$ is a semantic
translation, and $\fT$ is correct w.r.t.\ $\bR$.
\end{definition}

\begin{definition}{valid}
A translation $\fT$ is \emph{valid} up to $\asim$ iff
it is correct w.r.t.\ some semantic translation $\mathord{\bR} \subseteq \mathord{\asim}$.
\\
Language $\fL'$ is at least as \emph{expressive} as $\fL$ up to $\asim$ if a
translation valid up to $\asim$ from $\fL$ into $\fL'$ exists.
\end{definition}

\section{Correct = valid + congruence}\label{sec:valid-correct}

In \cite{vG12} the concept of a correct translation up to $\sim$ was defined, for 
$\sim$ a semantic equivalence on $\bZ$.\linebreak[3]
Here two valuations $\eta,\rho:\V\rightarrow\bZ$ are called \emph{$\sim$-equivalent},
$\eta\sim\rho$, if $\eta(X)\sim\rho(X)$ for each $X\in \V$.
In case there exists a $\val\in \bV$ for which there is no
$\sim$-equivalent $\val'\in \bV'$, there is no correct translation from
$\fL$ into $\fL'$ up to $\sim$.  Namely, the semantics of $\fL$
describes, among others, how any $\fL$-operator evaluates the argument
value $\val$, and this aspect of the language has no counterpart in $\fL'$.
Therefore, \cite{vG12} requires\vspace{-1.5ex}
\begin{equation}\label{related}
\forall \val\in \bV.~ \exists \val'\in \bV'.~ \val'\sim \val.
\end{equation}
This implies that for any valuation $\rho:\V\rightarrow\bV$ there is a
valuation $\eta:\V\rightarrow\bV'$ with $\eta\sim\rho$.

\begin{definitionc}{correct translation}{\cite{vG12}}
A translation $\fT$ from $\fL$ into $\fL'$ is \emph{correct up to $\sim$} iff
(\ref{related}) holds and\\ $\denote{\fT(\E)}_{\fL'}(\eta) \sim \denote{\E}_\fL(\rho)$ for all
$\E\in \IT_\fL$ and all valuations $\eta:\V\rightarrow\bV'$ and $\rho:\V\rightarrow\bV$ with
$\eta\sim\rho$.
\end{definitionc}
Note that this definition agrees completely with \df{correct}.
Requirement (\ref{related}) above corresponds to $\bR$ being a semantic translation in \df{correct}.

If a correct translation up to $\sim$ from $\fL$ into $\fL'$ exists, then $\sim$ must be a
congruence for $\fL$.

\begin{definition}{congruence}
An equivalence relation $\sim$ is a \emph{congruence} for a language
$\fL$ interpreted in  a semantic domain $\bV$ if
$\denote{E}_\fL(\nu)\sim\denote{E}_\fL(\rho)$
for any $\fL$-expression $E$ and any valuations
$\nu,\rho\!:\V\rightarrow \bV$ with $\nu \sim \rho$.\footnote{This is called a \emph{lean}
  congruence in \cite{vG17b}; in the presence of recursion, stricter congruence requirements are common.
  Those are not needed in this paper.}
\end{definition}

\begin{propositionc}{congruence}{\cite{vG12}}
$\!$If $\fT\!\!$ is a correct translation up to $\sim$ from $\fL$ into $\fL'\!\!$,
then $\sim$ is a congruence for $\fL\!\!$.
\end{propositionc}
The existence of a correct translation up to $\sim$ from $\fL$ into $\fL'$ does not imply
that $\sim$ is a congruence for $\fL'$. However, $\sim$ has the properties of a congruence for
those expressions of $\fL'$ that arise as translations of expressions of $\fL$, when
restricting attention to valuations into $\bU:=\{\val\in\bV'\mid \exists \val\in \bV.~\val'\sim \val\}$.
In \cite{vG12} this called a \emph{congruence for} $\fT(\fL)$.%
\begin{definition}{weak congruence}
Let $\fT: \IT_\fL \rightarrow \IT_{\fL'}$ be a translation from $\fL$ into $\fL'$.
An equivalence $\sim$ on $\bV'$ is a \emph{congruence for} $\fT(\fL)$
if $\denote{\fT(E)}_{\fL'}(\theta)\mathbin\sim\denote{\fT(E)}_{\fL'}(\eta)$
for any $E\mathbin\in\IT_\fL$ and $\theta,\eta\!:\!\V\!\!\!\rightarrow \!\bU$ with $\theta \mathbin\sim \eta$.%
\end{definition}

\begin{propositionc}{weak congruence}{\cite{vG12}}
If $\fT$ is a correct translation up to $\sim$ from $\fL$ into $\fL'$, then $\sim$ is a
congruence for $\fT(\fL)$.
\end{propositionc}
The following theorem tells that the notion of validity proposed in \sect{validity} can be seen as a
generalisation of the notion of correctness from \cite{vG12} that applies to
equivalences (and preorders) $\asim$ that are not necessarily congruences for $\fL$ or $\fT(\fL)$.

\begin{theorem}{valid correct}
A translation $\fT$ from $\fL$ into $\fL'$ is correct up to a semantic equivalence $\sim$ iff
it is valid up to $\sim$ and $\sim$ is a congruence for $\fT(\fL)$.
\end{theorem}

\begin{proof}
By Definitions~\ref{df:correct} and ~\ref{df:valid} any translation that is correct up to $\sim$
is surely valid up to $\sim$.

Suppose $\fT$ is valid up to $\sim$ and $\sim$ is a congruence for $\fT(\fL)$.
Then there is a semantic translation $\mathord{\bR}\subseteq \bV' \times \bV$ such that
$\mathord{\bR} \subseteq \mathord{\sim}$ and $\fT$ is correct w.r.t.\ $\bR$.
To establish that $\fT$ is correct up to $\sim$,
let $\E\in \IT_\fL$ and let $\eta:\V\rightarrow\bV'$ and $\rho:\V\rightarrow \bV$ be
valuations with $\eta\sim\rho$.
Let $\theta:\V\rightarrow\bV'$ be a valuation with $\theta \bR \rho$---it exists since $\bR$ is a
semantic translation.
Now $\theta \sim \rho \sim \eta$, using that $\mathord{\bR} \subseteq \mathord{\sim}$, so
$\theta,\eta:\V\rightarrow\bU$ and $\theta\sim\eta$. Hence
$\denote{\fT(\E)}_{\fL'}(\eta) \sim \denote{\fT(\E)}_{\fL'}(\theta) \sim \denote{\E}_\fL(\rho)$,
using that $\sim$ is a congruence for $\fT(\fL)$ and that $\fT$ is correct w.r.t.\ $\bR$.
\end{proof}

\section{A hierarchy of expressiveness preorders}\label{sec:hierarchy}

An equivalence or preorder $\asim$ on a class $\bZ$ is said to be \emph{finer},
\emph{stronger}, or \emph{more discriminating} than another equivalence or
preorder \mbox{${\scriptscriptstyle\bullet}\hspace{-1.15ex}\approx$} on $\bZ$ if $\val \asim \wal
\Rightarrow \val \mathrel{\mbox{${\scriptscriptstyle\bullet}\hspace{-1.15ex}\approx$}} \wal$ for all $\val,\wal\in\bZ$. 
\begin{observation}{hierarchy}
Let $\fT: \IT_\fL \rightarrow \IT_{\fL'}$ be a translation from $\fL$ into $\fL'$, and let
$\asim$ be finer than \mbox{${\scriptscriptstyle\bullet}\hspace{-1.15ex}\approx$}.
If $\fT$ is valid up to $\asim$, then it is also valid up to
\mbox{${\scriptscriptstyle\bullet}\hspace{-1.15ex}\approx$}.
\end{observation}
The quality of a translation depends on the choice of the equivalence or
preorder up to which it is valid. Any two languages are equally expressive up to
the universal equivalence, relating any two processes.
Hence, the equivalence or preorder needs to be chosen carefully to match the
intended applications of the languages under comparison. In general,
as shown by \obs{hierarchy}, using a finer equivalence or preorder yields a
stronger claim that one language can be encoded in another.
On the other hand, when separating two languages $\fL$ and $\fL'$ by
showing that $\fL$ \emph{cannot} be encoded in $\fL'$, a coarser
equivalence or preorder yields a stronger claim.

\begin{observation}{identity}
The identity is a valid translation up to any preorder from any
language into itself.
\end{observation}

\begin{theorem}{composition}
If valid translations up to $\asim$ exists from $\fL_1$ into $\fL_2$ and from $\fL_2$
into $\fL_3$, then there is a valid translation up to $\asim$ from $\fL_1$ into $\fL_3$.
\end{theorem}

\begin{proof}
For $i=1,2,3$ let $\denote{\ \ }_{\fL_i}:\IT_{\fL_i} \rightarrow ((\V\rightarrow\bV_i)\rightarrow\bV_i)$,
and for $k=1,2$ let $\fT_k:\IT_{\fL_k}\rightarrow\IT_{\fL_{k+1}}$ be valid translations up to
$\asim$ from $\fL_k$ to $\fL_{k+1}$, based on semantic translations
$\mathbin{\bR_k}\subseteq \bV_{\fL_{k+1}}\times\bV_{\fL_{k}}$. I will show that the translation
$\fT_2\circ\fT_1:\IT_{\fL_1}\rightarrow\IT_{\fL_3}$ from $\fL_1$ into $\fL_3$,
given by $\fT_2\circ\fT_1(E)=\fT_2(\fT_1(E))$, is valid up to $\asim$.
To this end, let $\mathord{\bR_2\circ\bR_1}\subseteq \bV_3\times\bV_1$ be the relation
given by $\val_3 \bR_2\circ\bR_1 \val_1$ iff \mbox{$\exists \val_2: \val_3 \bR_2 \val_2 \wedge \val_2
\bR_1 \val_1$}---\linebreak[1]it is again a semantic translation, and satisfies
$\mathord{\bR_2\circ\bR_1}\subseteq\mathord{\asim}$, using the transitivity of $\asim$.
Now let $\E\in \IT_{\fL_1}$, $\rho:\V\rightarrow\bV_1$ and $\theta:\V\rightarrow\bV_3$, with $\theta\bR_2\circ\bR_1\rho$.
Then there is a valuation $\eta:\V\rightarrow\bV_2$ with $\theta\bR_2\eta\bR_1\rho$.
Hence $\denote{\fT_2(\fT_1(\E))}_{\fL_3}(\theta) \asim
\denote{\fT_1(\E)}_{\fL_2}(\eta) \asim \denote{\E}_{\fL_1}(\rho)$,
so $\fT_2\circ\fT_1$ is correct w.r.t.\ $\bR_2\circ\bR_1$.
\end{proof}
\thm{composition} and \obs{identity} show that the relation
``being at least as expressive as up to $\asim$'' is a preorder on languages.

\section{Closed-term languages}\label{sec:closed-term}

The languages considered in this paper feature  \emph{variables}, \emph{operators} of \emph{arity}
$n\in\IN$, and/or other constructs. The set $\IT_\fL$ of $\fL$-expressions is inductively defined by:
\begin{itemize}
\item $X\in\IT_\fL$ for each variable $X\in\V$,
\item $f(E_1,\ldots,E_n)\in\IT_\fL$ for each $n$-ary operator $f$ and expressions $E_i\in\IT_\fL$ for $i=1,\ldots,n$,
\item and clauses for the other constructs, if any.
\end{itemize}
Examples of other constructs are the infinite summation operator $\sum_{i\in I}E_i$ of CCS, which
takes arbitrary many arguments, or the recursion construct $\mu X.E$, that has one argument, but
\emph{binds} all occurrences of $X$ in that argument.

In general a construct has a number (possibly infinite) of argument expressions and it may bind
certain variables within some of its arguments---the \emph{scope} of the binding.
An occurrence of a variable $X$ in an expression is \emph{bound} if it occurs within the scope of a
construct that binds $X$, and \emph{free} otherwise.

The semantics of such a language is given, in part, by a domain of
values $\bV$, and an interpretation of each $n$-ary operator $f$ of
$\fL$ as an $n$-ary operation $f^\bV: \bV^n\rightarrow \bV$ on $\bV$.
Using the equations $$\denote{X}_\fL(\rho) = \rho(X) \qquad \mbox{and} \qquad
\denote{f(E_1,\ldots,E_n)}_\fL(\rho) = f^\bV(\denote{E_1}_\fL(\rho),\ldots,\denote{E_n}_\fL(\rho))$$
this allows an inductive definition of the meaning $\denote{\E}_\fL$
of an $\fL$-expression $\E$.
Moreover, $\denote{E}_\fL(\rho)$ only depends on the restriction of $\rho$ to
the set $\fv(E)$ of variables occurring free in $\E$.

The set $\T_\fL \subseteq \IT_\fL$ of \emph{closed terms} of $\fL$ consists of those $\fL$-expressions $E\in\IT_\fL$
with $\fv(E)=\emptyset$. If $\p\in\T_\fL$ and $\bV\neq\emptyset$ then
$\denote{\p}_\fL(\rho)$ is independent of the choice of $\rho\!:\!\V\!\rightarrow\bV$, and hence
written $\denote{\p}_\fL$.

\begin{definition}{substitution}
A \emph{substitution} in $\fL$ is a partial function $\sigma:\V\rightharpoonup\IT_\fL$ from the
variables to the $\fL$-expressions. For a given $\fL$-expression $E\in\IT_\fL$,
$E[\sigma]\in\IT_\fL$ denotes the $\fL$-expression $E$ in which each free occurrence of a
variable $X\in\dom(\sigma)$ is replaced by $\sigma(X)$, while renaming bound variables in $E$
so as to avoid a free variable $Y$ occurring in an expression $\sigma(X)$ ending up being
bound in $E[\sigma]$.\footnote{In languages with multiple kinds of binding, such as for process
  variables $X\in \V$ and for names $x\in\N$ in the $\pi$-calculus below, all types of bound
  variables may be renamed.}\linebreak[3]
 A substitution is \emph{closed} if it has the form $\sigma:\V\rightarrow\T_\fL$.
\end{definition}
An important class of languages used in concurrency theory are the ones
where the distinction between syntax and semantic is effectively dropped by taking $\bV=\T_\fL$,
i.e.\ where the domain of values where the language is interpreted in consists of the closed terms
of the language. Here a valuation is the same as a closed substitution, and $\denote{E}_\fL(\rho)$
for $E\in\IT_\fL$ and $\rho:\V\rightarrow \T_\fL$ is defined to be $E[\rho]\in\T_\fL$.
I will call such languages \emph{closed-term} languages.

\section[Example: translating a synchronous into an asynchronous pi-calculus]
        {Example: translating a synchronous into an asynchronous $\pi$-calculus}\label{sec:asynpi}

As an illustration of the concepts introduced above, consider the
$\pi$-calculus as presented by Milner in \cite{Mi92}, i.e., the one of
Sangiorgi and Walker \cite{SW01book} without matching, $\tau$-prefixing, and choice.

Given a set of \emph{names } $\N$, the set $\IT_\pi$ of \emph{process expressions} or \emph{terms} $E$
of the calculus is given by  $$E ::= X ~~\mid~~ \textbf{0}  ~~\mid~~ \bar xy.E ~~\mid~~ x(z).E ~~\mid~~ E|E' ~~\mid~~ (\nu z)E ~~\mid~~ !E$$
with $x,y,z$ ranging over $\N$, and $X$ over $\V$, the set of \emph{process variables}.
Process variables are not considered in \cite{SW01book}, although they are common in languages like
CCS \cite{Mi90ccs} that feature a recursion construct. Since process variables form a central part
of my notion of a valid or correct translation, here they have simply been added. This works generally.
In \sect{compositionality} I show that for the purpose of accessing whether one language is as
expressive as another, translations between them can be assumed to be compositional.
This important result would be lost if process variables were dropped from the language.
In that case compositionality would need to be stated as a separate requirement for valid translations.

Closed process expressions are called \emph{processes}.
The $\pi$-calculus is usually presented as a closed-term language, in that the semantic value
associated with a closed term is simply itself. Yet, the real semantics is given by a reduction
relation between processes, defined below.

\begin{definition}{structural congruence}
An occurrence of a name $z$ in $\pi$-calculus process $P\in\T_\pi$ is \emph{bound} if it occurs
within a subexpression $x(z).P'$ or $(\nu z)P'$ of $P$; otherwise it is \emph{free}.
Let $\n(P)$ (resp.\ $\bn(P)$) be the set of names occurring (bound) in $P\in \T_\pi$.

\emph{Structural congruence}, $\equiv$, is the smallest congruence relation on processes satisfying
\[
\begin{array}[b]{r@{~\equiv~}l@{\qquad}r@{~\equiv~}l@{\qquad}r@{~\equiv~}l@{~~\mbox{if}~ w\notin\n(P)}}
P_1 | (P_2 | P_3) & (P_1 | P_2) | P_3 &
!P & P | !P &
(\nu w) (P | Q) & P | (\nu w)Q \\
P_1 | P_2 & P_2 | P_1 &
(\nu z) \textbf{0} & \textbf{0} &
x(z).P & x(w).P\subs{w}{z} \\
P | \textbf{0} & P &
(\nu z)(\nu w)P & (\nu w)(\nu z) P &
(\nu z)P & (\nu w)P\subs{w}{z}\\
\end{array}.
\]
Here $P\subs{w}{z}$ denotes the process obtained by replacing each free occurrence of $z$ in $P$ by $w$.

\end{definition}

\begin{definition}{reduction}
The \emph{reduction relation}, ${\rightarrow}\subseteq \T_\pi \times \T_\pi$, is generated by the
following rules.
\[
\frac{z\notin\bn(Q)}{\bar xz.P | x(y).Q \rightarrow P|Q\subs{z}{y}} \qquad
\frac{P \rightarrow P'}{P|Q \rightarrow P'|Q} \qquad
\frac{P \rightarrow P'}{(\nu z)P \rightarrow (\nu z)P'} \qquad
\frac{Q \equiv P \quad P \rightarrow P' \quad P' \equiv Q'}{Q \rightarrow Q'}
\]
\end{definition}
Let $\Longrightarrow$ be the reflexive and transitive closure of $\rightarrow$.
The observable behaviour of $\pi$-calculus processes is often stated in terms of the outputs
they can produce (abstracting from the value communicated on an output channel).

\begin{definition}{barbs}
Let $x\in\N$. A process $P$ has a \emph{strong output barb} on $x$, notation $\sbarb{P}{\bar x}$, if
$P$ can perform an output action $\bar xz$. This is defined inductively:
\[
\sbarb{(\bar xz.(P))}{\bar x} \qquad
\frac{\sbarb{P}{\bar x}}{\sbarb{(P|Q)}{\bar x}} \qquad
\frac{\sbarb{Q}{\bar x}}{\sbarb{(P|Q)}{\bar x}} \qquad
\frac{\sbarb{P}{\bar x} \quad x\neq z}{\sbarb{((\nu z)P)}{\bar x}} \qquad
\frac{\sbarb{P}{\bar x}}{\sbarb{(!P)}{\bar x}}
\]
A process $P$ has a \emph{weak output barb} on $x$, notation $\wbarb{P}{\bar x}$, if
there is a $P'$ with $P \Longrightarrow \wbarb{P'}{\bar x}$.
\end{definition}
A common semantic equivalence applied in the $\pi$-calculus is \emph{weak barbed congruence} \cite{MilS92,SW01,SW01book}.
\begin{definition}{barbed congruence}
\emph{Weak (output) barbed bisimilarity} is the largest symmetric relation ${\wbb}\subseteq \T_\pi \times \T_\pi$
such that
\begin{itemize}
\item $P \wbb Q$ and $\sbarb{P}{\bar x}$ implies $\wbarb{Q}{\bar x}$, and
\item $P \wbb Q$ and $P \Longrightarrow P'$ implies $Q \Longrightarrow Q'$ for some $Q'$ with $P' \wbb Q'$.
\end{itemize}
\emph{Weak barbed congruence}, $\cong^c$, is the largest congruence included in $\wbb$.
\end{definition}
Often \emph{input barbs}, defined similarly, are included in the definition of weak barbed bisimilarity \cite{SW01book}.
This is known to induce the same notion of weak barbed congruence \cite{SW01book}.
Another technique for defining weak barbed congruence is to use a barb, or set of barbs, external to
the language under investigation, that are added to the language as constants, similar to the theory
of testing of Hennessy and De Nicola~\cite{DH84}---see \sect{barbed} for details. This method is
useful for  languages with a reduction semantics that do not feature a clear notion of barb, or
where there is ambiguity in which barbs should be counted and which not, or for comparing languages
with different kinds of barb.

\begin{example}{asynchronous congruence}
$\bar x z.\textbf{0} \not\cong^c (\nu u)(\bar x u .\textbf{0}| u(v). \bar v z. \textbf{0})$.
For let $E := X | x(u).\bar u v.\textbf{0}$ with $\zeta(X)=(\nu u)(\bar x u.\textbf{0} | u(v). \bar v z. \textbf{0})$
and $\rho(X)=\bar x z.\textbf{0}$.
Then $E[\zeta] \rightarrow (\nu u)\big(u(v). \bar v z. \textbf{0} | \bar u v.\textbf{0}\big) \rightarrow \sbarb{(\bar v z. \textbf{0})}{\bar v}$
but $\nwbarb{(E[\rho])}{\bar v}$.
\end{example}

\noindent
The asynchronous $\pi$-calculus, as introduced by Honda \& Tokoro in \cite{HT91} and by Boudol in
\cite{Bo92}, is the sublanguage $\rm a\pi$ of the fragment $\pi$ of the $\pi$-calculus presented above where all
subexpressions $\bar x y.E$ have the form $\bar x y.\textbf{0}$. \emph{Asynchronous barbed congruence},
$\cong^c_{\rm a}$, is the largest congruence \emph{for the asynchronous $\pi$-calculus} included in $\wbb$.
Since $a\pi$ is a sublanguage of $\pi$, $\cong^c_{\rm a}$ is at least as coarse an equivalence as $\cong^c$,
i.e.\ ${\cong^c} \subseteq {\cong^c_{\rm a}}$. The inclusion is strict, since 
$!x(z).\bar x z.\textbf{0} \cong^c_{\rm a} \textbf{0}$, yet $!x(z).\bar x z.\textbf{0} \not\cong^c \textbf{0}$ \cite{SW01book}.
Since all expressions used in \ex{asynchronous congruence} belong to $\rm a\pi$, one even has 
$\bar x z.\textbf{0} \not\cong^c_{\rm a} (\nu u)(\bar x u .\textbf{0}| u(v). \bar v z. \textbf{0})$.

Boudol \cite{Bo92} defined a translation $\fT$ from $\pi$ to $\rm a\pi$ inductively as follows:
\[
\begin{array}{rcll}
\fT(X)            &=& X & \mbox{for}~X\in\V\\
\fT(\textbf{0})   &=& \textbf{0} \\
\fT(\bar{x}z.P)   &=& (u)(\bar{x}u | u(v).(\bar{v}z | \fT(P))) & \mbox{choosing}~u,v \notin\n(P),~ u\neq v\\
\fT(x(y).P)       &=& x(u).(v)(\bar{u}v | v(y).\fT(P)) & \mbox{choosing}~u,v \notin\n(P),~ u\neq v \\
\fT(P | Q)        &=& (\fT(P) | \fT(Q)) \\
\fT(!P)           &=& \ ! \fT(P) \\
\fT((\nu x) P)    &=& (\nu x) \fT(P)
\end{array}
\]
\ex{asynchronous congruence} shows that $\fT$ is not valid up to $\cong^c$.
In fact, it is not even valid up to $\cong^c_{\rm a}$.
However, as shown in \cite{LMGG18}, it is valid up to $\wbb$.
Since $\wbb$ is not a congruence (for $\pi$ or $\rm a\pi$) it is not correct up to $\wbb$.

\section{Congruence closure}\label{sec:congruence closure}

\begin{definition}{1-hole congruence}
An equivalence relation $\sim$ is a \emph{1-hole congruence} for a language
$\fL$ interpreted in  a semantic domain $\bV$ if
$\denote{E}_\fL(\nu)\sim\denote{E}_\fL(\rho)$
for any $\fL$-expression $E$ and any valuations
$\nu,\rho:\V\rightarrow \bV$ with $\nu \sim^1 \rho$.
Here $\nu,\rho$ are \emph{$\sim^1$-equivalent}, $\nu\sim^1\rho$,
if $\nu(X)\sim\rho(X)$ for some $X\in \V$ and $\nu(Y)=\rho(Y)$ for all variables $Y\neq X$.
\end{definition}
An \emph{$n$-hole congruence} for any finite $n\in\IN$ can be defined in the same vain, and it is
well known and easy to check that a 1-hole congruence $\sim$ is also an $n$-hole congruence, for any $n\in\IN$.
However, in the presence of operators with infinitely many arguments, a 1-hole congruence need not
be a congruence.
\begin{example}{1-hole congruence}
Let $\bV$ be $(\IN\times\IN) \cup\{\infty\}$, with the well-order $\leq$ on $\bV$ inherited lexicographically from the default
order on $\IN$ and $\infty$ the largest element. So $(n,m) \leq (n',m')$ iff $n\leq n' \vee (n=m \wedge m\leq m')$. Consider the
language $\fL$ with constants $0$, $1$ and $(1)$, interpreted in $\bV$ as $(0,0)$, $(1,0)$ and
$(0,1)$, respectively, the binary operator $+$, interpreted by $(n_1,m_1) +^{\!\bV} (n_2,m_2) =
(n_1\mathord+n_2,m_1\mathord+m_2)$ and $\infty+E=E+\infty=\infty$,
and the construct $\sup (E_i)_{i\in I}$ that takes any number of arguments (dependent on the set of
the index sets $I$). The interpretation of $\sup$ in $\bV$ is to take the supremum of its
arguments w.r.t.\ the well-order $\leq$. In case $\sup$ is given finitely many arguments, it
simply returns the largest. However $\sup((n,i))_{i\in\INs}=(n\mathord+1,0)$.

Now let the equivalence relation $\sim$ on $\bV$ be defined by $(n,m)\sim(n',m')$ iff $n=n'$,
leaving $\infty$ in an equivalence class of its own.
This relation is a 1-hole congruence on $\fL$. Hence, it is also a 2-hole congruence, so one has\vspace{-3pt}
$$\big((n_1,m_1)\sim (n_1',m_1') \wedge (n_2,m_2)\sim (n_2',m_2')\big) \Rightarrow
(n_1,m_1)+ (n_2,m_2) \sim (n'_1,m'_1)+ (n'_2,m'_2).$$
Yet it fails to be a congruence: $(n,i) \sim (n,0)$ for all $i\in\IN$, but
$$(n\mathord+1,0)=\sup((n,i))_{i\in\INs}\not\sim \sup((n,0))_{i\in\INs}=(n,0).$$
\end{example}
It is well known and easy to check that the collection of equivalence relations on any domain $\bV$, ordered by
inclusion, forms a complete lattice---namely the intersection of arbitrary many equivalence
relations is again an equivalence relation. Likewise, the collection of 1-hole congruences for $\fL$
is also a complete lattice, and moreover a complete sublattice of the complete lattice of
equivalence relations on $\bV$.  The latter implies that for any collection $C$ of 1-hole congruence
relations, the least equivalence relation that contains all elements of $C$ (exists and) happens to
be a 1-hole congruence relation. Again, this is a property that is well known \cite{Gr10} and easy
to prove. It follows that for any equivalence relation $\sim$ there exists a largest 1-hole
congruence for $\fL$ contained in $\sim$. I will denote this 1-hole congruence by $\sim^{1c}_\fL $,
and call it the \emph{congruence closure} of $\sim$ w.r.t.\ $\fL$.
One has $\val_1\sim^{1c}_\fL\val_2$ for $\val_1,\val_2\in\bV$ iff
$\denote{E}_\fL(\nu)\sim\denote{E}_\fL(\rho)$
for any $\fL$-expression $E$ and any valuations
$\nu,\rho:\V\rightarrow \bV$ with
$\nu(X)=\val_1$ and $\rho(X)=\val_2$ for some $X\in \V$ and $\nu(Y)=\rho(Y)$ for all variables $Y\neq X$.
Such results do not generally hold for congruences.
\begin{example}{no largest congruence}
Continue \ex{1-hole congruence}, but skipping the operator $+$.
Let $\sim_k$ be the equivalence on $\bV$ defined by $(n,m)\sim_k(n',m')$ iff $n=n' \wedge (m=m' \vee m,m'\leq k)$.
It is easy to check that all $\sim_k$ for $k\in\IN$ are congruences on the reduced $\fL$, and
contained in $\sim$. Yet their least upper bound (in the lattice of equivalence relations on $\bV$)
is $\sim$, which is not a congruence itself. In particular, there is no largest congruence contained in $\sim$.
\end{example}
When dealing with languages $\fL$ in which all operators and other constructs have a finite arity, so that
each $E\in\IT_\fL$ contains only finitely many variables, there is no difference between a
congruence and a 1-hole congruence, and thus $\sim^{1c}_\fL$ is a congruence relation for any equivalence $\sim$.
I will apply the theory of expressiveness presented in this paper also to languages like CCS that
have operators (such as $\sum_{i\in I}E_i$) of infinite arity. However, in all such cases I'm
currently aware of, the relevant choices of $\fL$ and $\sim$ have the property that $\sim^{1c}_\fL$ is
in fact a congruence relation. As an example, consider weak bisimilarity \cite{Mi90ccs}. This equivalence relation
fails to be a congruence for $\sum$. However, the coarsest 1-hole congruence contained in this
relation, often called \emph{rooted} weak bisimilarity, happens to be a congruence. In fact, when
congruence-closing weak bisimilarity w.r.t.\ the binary sum, the result \cite{vG05e} is also a
congruence for the infinitary sum, as well as for all other operators of CCS \cite{Mi90ccs}.

\begin{definition}{1-hole congruence for image}
Let $\fT$ be a translation from $\fL$ into $\fL'$.
A subset $\bW$ of $\bV'$ is \emph{closed} under $\fT(\fL)$ if $\denote{\fT(E)}(\eta) \in \bW$
for any expression $E\in\IT_\fL$ and valuation $\eta:\V \rightarrow \bW$.
An equivalence $\sim$ on $\bW$ is a \emph{congruence} (respectively \emph{1-hole  congruence})
for $\fT(\fL)$ on $\bW$ if for any $E\in\IT_\fL$ and $\theta,\eta:\V\rightarrow \bW$ with $\theta
\sim \eta$ (respectively $\theta\sim^1\eta$) one has
$\denote{\fT(E)}_{\fL'}(\theta)\sim\denote{\fT(E)}_{\fL'}(\eta)$.
\end{definition}

\begin{proposition}{closed}
Let $\fT$ be a translation from $\fL$ into $\fL'$ that is correct w.r.t.\ a semantic translation
$\mathord{\bR}\subseteq\bV'\times\bV$. Let $\mathord{\bR}(\bV):=\{\val'\in\bV' \mid \exists \val\in\bV.~\val'\bR \val\}$.
Then $\mathord{\bR}(\bV)$ is closed under $\fT(\fL)$.
\end{proposition}

\begin{proof}
Let $E\in\IT_\fL$ and $\eta:\V \rightarrow \mathord{\bR}(\bV)$. Take $\rho:\V \rightarrow\bV$ with
$\rho \bR \eta$. Then $\denote{\fT(\E)}_{\fL'}(\eta) \bR \denote{\E}_\fL(\rho)$.
Since $\denote{\E}_\fL(\rho)\in\bV$ one has $\denote{\fT(\E)}_{\fL'}(\eta)\in\mathord{\bR}(\bV)$.
\end{proof}

\begin{proposition}{1-hole congruence for image}
Let the translation $\fT$ from $\fL$ into $\fL'$ be correct w.r.t.\ the semantic translation
$\mathord{\bR}\subseteq\mathord{\sim}$.
Then $\sim$ is a (1-hole) congruence for $\fL$ iff it is a (1-hole) congruence for $\fT(\fL)$ on $\mathord{\bR}(\bV)$.
\end{proposition}

\begin{proof}
First suppose $\sim$ is a congruence for $\fL$.
Let $E\mathbin\in\IT_\fL$ and $\theta,\eta:\V\rightarrow \mathord{\bR}(\bV)$ with $\theta \sim \eta$.
By the definition of $\mathord{\bR}(\bV)$ there are valuations $\nu,\rho:\V\rightarrow\bV$ with 
$\theta \bR \nu$ and $\eta \bR \rho$. Now $\nu\sim\theta\sim\eta\sim\rho$, so
$$\denote{\fT(E)}_{\fL'}(\theta) \bR \denote{\E}_\fL(\nu) \sim
\denote{\E}_\fL(\rho) \bR^{-1} \denote{\fT(E)}_{\fL'}(\eta)$$
and hence $\denote{\fT(E)}_{\fL'}(\theta) \sim \denote{\fT(E)}_{\fL'}(\eta)$.
The other direction proceeds in the same way.

Now suppose $\sim$ is a 1-hole congruence for $\fL$.
Let $E\mathbin\in\IT_\fL$ and $\theta,\eta:\V\rightarrow \mathord{\bR}(\bV)$ with $\theta \sim^1 \eta$.
Then $\theta(X)\sim\eta(X)$ for some $X\in\V$ and $\theta(Y)=\eta(Y)$ for all $Y\neq X$.
So there must be $\nu,\rho:\V\rightarrow\bV$ with $\theta \bR \nu$, $\eta \bR \rho$ and
$\nu(Y)=\rho(Y)$ for all $Y\neq X$. Since $\nu(X)\sim\theta(X)\sim\eta(X)\sim\rho(X)$ it follows
that $\nu\sim^1\rho$. The conclusion proceeds as above, and the other direction goes likewise.
\end{proof}
The requirement of being a congruence for $\fT(\fL)$ on $\mathord{\bR}(\bV)$ is slightly weaker than that of
being a congruence for $\fT(\fL)$---cf.\ \df{weak congruence}---for it proceeds by restricting
attention to valuations into $\mathord{\bR}(\bV)\subseteq \bU$.

\begin{example}{congruence on T(L)}
Let $\fL$ be a language with two constants \textbf{Yes} and \textbf{No} and a unary operator $\neg$.
Its semantics is given by $\bV=\{0,1\}$, $\textbf{Yes}^\bV=1$, $\textbf{No}^\bV=0$ and $\neg^\bV(b)=1{-}b$.
Let $\fL'$ be the extension of $\fL$ with the constant $\top$ and the semantic value $\top$---so
$\bV'=\{0,1,\top\}$---with $\top^{\bV'}=\top$ and $\neg^{\bV'}(\top)=\top$.
Let $\sim$ be the semantic equivalence on $\bV' \supseteq \bV$ given by $0\not\sim 1 \sim \top$.
Then the identity function $\cal I$ is a translation from $\fL$ into $\fL'$ that is valid up to $\sim$.
This is witnessed by the semantic translation ${\bR} :=\{(0,0),(1,1)\} \subseteq {\sim}$.
The relation $\sim$ is a congruence for $\fL$. In line with \pr{1-hole congruence for image} it also
is a congruence for ${\cal I}(\fL)$ on ${\bR}(\bV)=\{0,1\}$. However it fails to be a congruence
for ${\cal I}(\fL)$ (on $\bU=\{0,1,\top\}$). So by \pr{weak congruence} the transition $\cal I$ is not correct up
to $\sim$. Indeed, for $\rho:\V\rightarrow\bV$ a valuation that sends $X\in \V$ to $1$ and
$\eta:\V\rightarrow\bV'$ a valuation that sends $X$ to $\top$, assuming $\rho(Y)=\eta(Y)$ for other
variables $Y$, one has $\rho \sim \eta$, yet
$\denote{{\cal I}(\neg X)}_{\fL'}(\eta) = \denote{\neg X}_{\fL'}(\eta) = \neg^{\bV'}(\eta(X)) = \top
\not\sim 0 = \neg^{\bV}(\rho(X)) =\denote{\neg X}_{\fL}(\rho)$.
\end{example}

\section{A congruence closure property for valid translations}\label{sec:ccproperty}

In many applications, semantic values in the domain of interpretation of a language $\fL$ are only
meaningful up to a semantic equivalence $\sim^c$, and the intended semantic domain could just as
well be seen as the set of $\sim^c$-equivalence classes of values. For this purpose it is essential
that $\sim^c$ is a congruence for $\fL$. Often $\sim^c$ is the congruence closure of a coarser semantic
equivalence $\sim$, so that two values end up being identified iff they are $\sim$-equivalent in
every context. An example of this occurred in \sect{asynpi}, with $\wbb$ in the r\^ole of $\sim$ and
$\cong^c$ in the r\^ole of $\sim^c$. Now \thm{congruence closure of translation}, contributed in
this section, says that if a translation from $\fL$ into $\fL'$ is valid up to $\sim$, then it is
even valid up to an equivalence \plat{$\sim^{1c}_{\fL\!,\bR}$} that extends $\sim^c$ from $\bV$ to a
subdomain $\bW$ of $\bV'$ that suffices for the interpretation of translated expressions from $\fL$.
This equivalence \plat{$\sim^{1c}_{\fL\!,\bR}$}
coincides with the congruence closure of $\sim$ on $\fL$, as well as on $\fT(\fL)$,
and melts each equivalence class of $\bV$ with exactly one of $\bW$, and vice versa.
\\[1ex]
Let $\fL$ and $\fL'$ be languages with
$\denote{\ \ }_{\fL}:\IT_{\fL} \rightarrow ((\V\rightarrow\bV)\rightarrow\bV)$
and
$\denote{\ \ }_{\fL'}:\IT_{\fL'} \rightarrow ((\V\rightarrow\bV')\rightarrow\bV')$.
In this section I assume that $\bV\cap\bV'=\emptyset$. To apply the results to the general case,
just adapt $\fL'$ by using a copy of $\bV'$---any preorder $\asim$ on $\bV\cup\bV'$ extends to this
copy by considering each copied element $\asim$-equivalent to the original.

\begin{definition}{canonical equivalence}
Given any semantic translation $\bR$, let $\mathord{\equiv_{\bR}} \subseteq (\bV\cup\bV')^2$
be the smallest equivalence relation on $\bV\cup\bV'$ containing $\bR$.
\end{definition}

\begin{theorem}{canonical equivalence}
If a translation $\fT$ is correct w.r.t.\ the semantic translation $\bR$, then $\equiv_{\bR}$
is a 1-hole congruence for $\fL$.
\end{theorem}
\begin{proof}
Let $E\in\IT_\fL$ and $\nu,\rho:\V\rightarrow\bV$ with $\nu\equiv_{\bR}^1\rho$.
I have to show that $\denote{E}_\fL(\nu)\equiv_{\bR}\denote{E}_\fL(\rho)$.

Let $X\in\V$ be such that $\nu(X)\equiv_{\bR}\rho(X)$ and $\nu(Y)=\rho(Y)$ for all $Y\neq X$.
Then, for some $n\geq 0$ there are $\val_0,\ldots,\val_n\in\bV$ and $\wal_1,\ldots,\wal_n\in\bV'$ with
$\nu(X)=\val_0 \bR^{-1} \wal_1 \bR \val_1 \bR^{-1} \wal_2 \bR \val_2
\bR^{-1} \ldots \bR \val_n = \rho(X)$.
For $i=0,\ldots,n$ let $\rho_i:\V\rightarrow\bV$ be given by $\rho_i(X)=\val_i$ and $\rho_i(Y)=\rho(Y)$ for $Y\neq X$,
and for $i=1,\ldots,n$ let $\eta_i:\V\rightarrow\bV'$ be given by $\eta_i(X)=\wal_i$ and
$\eta_i(Y)=\eta(Y)$ for $Y\neq X$, for some $\eta:\V\rightarrow\bV'$ with $\eta\bR\rho$.
Then $\nu=\rho_0 \bR^{-1} \eta_1 \bR \rho_1 \bR^{-1}
\eta_2 \bR \rho_2 \bR^{-1} \ldots \bR \rho_n = \rho$.
Hence $\denote{E}_\fL(\rho_0) \bR^{-1} \denote{\fT(\E)}_{\fL'}(\eta_1)
\bR \denote{E}_\fL(\rho_1) \bR^{-1} \denote{\fT(\E)}_{\fL'}(\eta_2) \bR
\denote{E}_\fL(\rho_2) \bR^{-1} \cdots \bR \denote{E}_\fL(\rho_n)$.
Thus $\denote{E}_\fL(\nu)\equiv_{\bR}\denote{E}_\fL(\rho)$.
\end{proof}
By \pr{1-hole congruence for image} $\equiv_{\bR}$
also is a 1-hole congruence for $\fT(\fL)$ on $\mathord{\bR}(\bV)$.
Only the subset $\mathord{\bR}(\bV)$ of $\bV'$ matters for the purpose
of translating $\fL$ into $\fL'$. On $\bV'\setminus\mathord{\bR}(\bV)$ the equivalence
$\equiv_{\bR}$ is the identity.

\begin{theorem}{congruence closure of translation}
Let $\fT$ be a translation from a language $\fL$, with semantic domain $\bV$,
into a language $\fL'$, with domain $\bV'$, that is valid up to a semantic equivalence $\sim$.
Then $\fT$ is even valid up to a semantic equivalence \plat{$\sim^{1c}_{\fL\!,\bR}$}, contained in $\sim$, such that
(1) the restriction of \plat{$\sim^{1c}_{\fL\!,\bR}$} to $\bV$ is the largest 1-hole congruence for $\fL$ contained in $\sim$,
(2) the set $\bW := \{\val\in\bV'\mid \exists \val\in\bV.~\val'\sim^{1c}_{\fL\!,\bR} \val\}$ is closed under $\fT(\fL)$,
and (3) the restriction of \plat{$\sim^{1c}_{\fL\!,\bR}$} to $\bW$
is the largest 1-hole congruence for $\fT(\fL)$ on $\bW$ that is contained in $\sim$.
\end{theorem}
\begin{proof}
By assumption the translation $\fT$ from $\fL$ into $\fL'$
is correct w.r.t.\ a semantic translation $\mathord{\bR} \subseteq \mathord{\sim}$.
Let $\sim^{1c}_{\fL\!,\bR}$, the \emph{congruence closure} of $\sim$ w.r.t.\ $\fL$ and $\bR$, be
the binary relation on $\bV \cup \bV'$ defined by \plat{$\wal_1 \sim^{1c}_{\fL\!,\bR} \wal_2$} iff
\plat{$\wal_1 \equiv_{\bR} \val_1 \sim^{1c}_\fL \val_2 \equiv_{\bR} \wal_2$} for some $\val_1,\val_2\in\bV$.
Here \plat{$\sim^{1c}_\fL$} is the largest 1-hole congruence for $\fL$, defined on $\bV$, that is contained in $\sim$.

Since $\mathord{\equiv_{\bR}}$ is a 1-hole congruence for $\fL$
contained in $\sim$ (by \thm{canonical equivalence}), and $\sim^{1c}_\fL$ is the
largest 1-hole congruence for $\fL$ contained in $\sim$,
one has $\val \equiv_{\bR} \wal\Rightarrow \val \sim^{1c}_\fL \wal$ for all $\val,\wal\in\bV$.
From this it follows that $\sim^{1c}_{\fL\!,\bR}$ is transitive. As $\equiv_{\bR}$ and
\plat{$\sim^{1c}_\fL$} are reflexive and symmetric, so is \plat{$\sim^{1c}_{\fL\!,\bR}$}. Thus, \plat{$\sim^{1c}_{\fL\!,\bR}$}
is an equivalence relation.
Since \plat{$\sim^{1c}_\fL$} and $\bR$, and hence also $\equiv_{\bR}$, are contained in $\sim$, so is $\sim^{1c}_{\fL\!,\bR}$.
Moreover, \plat{$\mathord{\bR} \subseteq \mathord{\equiv_{\bR}} \subseteq \mathord{\sim^{1c}_{\fL\!,\bR}}$},
so $\fT$ is valid up to \plat{$\sim^{1c}_{\fL\!,\bR}$}.
It remains to check properties (1)--(3).

Let $E\in \IT_\fL$ be an $\fL$-expression and $\nu,\rho:\V\rightarrow \bV$ be valuations with
$\nu\sim^1 \rho$. Then there are valuations $\nu',\rho':\V\rightarrow\bV$ with
\plat{$\nu \equiv_{\bR}^1 \nu' \mathrel{(\sim^{1c}_\fL)^1} \rho' \equiv_{\bR}^1 \rho$}.
Since $\equiv_{\bR}$ and $\sim^{1c}_\fL$ are 1-hole congruences, it follows that 
$\denote{E}_\fL(\nu)\equiv_{\bR}\denote{E}_\fL(\nu')\sim^{1c}_\fL\denote{E}_\fL(\rho')\equiv_{\bR}
\denote{E}_\fL(\rho)$, so $\denote{E}_\fL(\nu)\sim^{1c}_{\fL\!,\bR}\denote{E}_\fL(\rho)$. Thus
\plat{$\sim^{1c}_{\fL\!,\bR}$} is a 1-hole congruence for $\fL$ on $\bV$. Since \plat{$\sim^{1c}_\fL$} is the
largest 1-hole congruence for $\fL$ contained in $\sim$, it follows that, restricted to $\bV$,
${\sim^{1c}_{\fL\!,\bR}} \subseteq {\sim^{1c}_\fL}$. By the reflexivity of $\equiv_{\bR}$ one moreover has
${\sim^{1c}_{\fL}} \subseteq {\sim^{1c}_{\fL\!,\bR}}$, so \plat{${\sim^{1c}_{\fL\!,\bR}} = {\sim^{1c}_\fL}$},
i.e.\ (1) holds.

Let $\bW := \{\wal\in\bV'\mid \exists \val\in\bV.~\wal\sim^{1c}_{\fL\!,\bR} \val\}$.
By definition, $\bW = {\bR}(\bV)$. So by \pr{closed} $\bW$ is closed under $\fT(\fL)$,
i.e.\ (2) holds.

By \pr{1-hole congruence for image} \plat{$\sim^{1c}_{\fL\!,\bR}$}
is a 1-hole congruence for $\fT(\fL)$ on $\bW = \mathord{\bR}(\bV)$, contained in $\sim$.
Let $\approx$ be any other 1-hole congruence on $\bW$ contained in $\sim$.
Define the relation $\approx_{\bR}$ on $\bV\cup\bW$ by \plat{$\val_1 \approx_{\bR} \val_2$} iff
\plat{$\val_1 \equiv_{\bR} \wal_1 \approx \wal_2 \equiv_{\bR} \val_2$} for some
$\wal_1,\wal_2\in{\bW}$, and let $\approx_{\bR}^*$ be its transitive closure.
Then $\approx_{\bR}^*$ is an equivalence relation on $\bV\cup\bW$.
Since $\equiv_{\bR}$ and $\approx$ are 1-hole congruences for $\fT(\fL)$ on $\bW$, by the same
reasoning as above also \plat{$\approx_{\bR}^*$} is a 1-hole congruence for $\fT(\fL)$ on $\bW$.
As \plat{$\mathord{\bR} \subseteq \mathord{\equiv_{\bR}} \subseteq {\approx_{\bR}^*}$},
by \pr{1-hole congruence for image}, \plat{$\approx_{\bR}^*$} is a 1-hole congruence for $\fL$.
Since \plat{$\sim^{1c}_\fL$} is the largest 1-hole congruence for $\fL$ contained in $\sim$, it
follows that, restricted to $\bV$, ${\approx_{\bR}^*} \subseteq {\sim^{1c}_\fL}$.

For each $\wal_1,\wal_2\in\bW$ with $\wal_1 \approx \wal_2$ there are $\val_1,\val_2\in\bV$ with
\plat{$\val_1 \equiv_{\bR} \wal_1 \approx \wal_2 \equiv_{\bR} \val_2$}. So $\val_1 \approx_{\bR} \val_2$
and hence $\val_1 \sim^{1c}_\fL \val_2$, and thus \plat{$\wal_1 \sim^{1c}_{\fL\!,\bR} \wal_2$}, by definition.
So ${\approx} \subseteq {\sim^{1c}_{\fL\!,\bR}}$ and (3) holds.
\end{proof}
Note that each equivalence class of $\sim^{1c}_{\fL\!,\bR}$ on $\bV \cup \bW$ melts an equivalence
class of $\sim^{1c}_{\fL\!,\bR}$ on $\bV$ with one of \plat{$\sim^{1c}_{\fL\!,\bR}$} on $\bW$.
Moreover, on $\bV$ the relation is completely determined by $\fL$ and $\sim$.
The following example shows that in general the whole relation \plat{$\sim^{1c}_{\fL\!,\bR}$} is not
completely determined by $\fL$ and $\sim$.

\begin{example}{congruence closure not determined by L and sim}
Let $\fL$ be a language with two constants \textbf{Yes} and \textbf{No} and a binary operator \textbf{same}.
Its semantics is given by $\bV=\{0,1,\top,\bot\}$, $\textbf{Yes}^\bV=1$, $\textbf{No}^\bV=0$ and
$\textbf{same}^\bV(x,y)=1$ if $x=y$, while $\textbf{same}^\bV(x,y)=0$ if $x\neq y$.
The semantic values $\top$ and $\bot$ are not denotable as the interpretation of closed terms.
Let $\fL'$ be an exact copy of $\fL$, except that the semantic values are primed.
Let $\sim$ be the semantic equivalence on $\bV \cup \bV'$ given by
$0 \sim 0' \not\sim 1 \sim 1' \not\sim \top \sim \top' \sim \bot \sim \bot'$.
Then the identity function $\cal I$ is a translation from $\fL$ into $\fL'$ that is valid up to $\sim$.
This is witnessed by the semantic translation ${\bR}^\dagger :=\{(0',0),(1',1),(\top',\bot),(\bot',\top)\} \subseteq {\sim}$,
or, alternatively, by the semantic translation ${\bR} :=\{(0',0),(1',1),(\top',\top),(\bot',\bot)\} \subseteq {\sim}$.
Upon taking the congruence closure $\sim^1_\fL$, the four semantic values of $\bV$ become
inequivalent. Yet, there are two candidates for \plat{$\sim^{1c}_{\fL\!,\bR}$}, namely $\bR^\dagger$ and $\bR$.
\end{example}

\begin{corollary}{congruence closure of translation}
Let $\fT$ be a translation from a language $\fL$, with semantic domain $\bV$,
into a language $\fL'$, with domain $\bV'$, valid up to a semantic equivalence $\sim$,
and suppose the congruence closure $\sim_\fL^1$ of $\sim$ w.r.t.\ $\fL$ is in fact a congruence.
Then $\fT$ is correct up to the equivalence \plat{$\sim^{1c}_{\fL\!,\bR}$} described in
\thm{congruence closure of translation}.
\end{corollary}
\begin{proof}
By the proof of \thm{congruence closure of translation} $\fT$ is correct w.r.t.\ the semantic
translation \plat{${\bR} \subseteq {\sim^{1c}_{\fL\!,\bR}}$}, the restriction of \plat{$\sim^{1c}_{\fL\!,\bR}$} to $\bV$
equals $\sim_\fL^1$, and ${\bR}(\bV)=\bW$. Using that $\sim_\fL^1$ is a congruence for $\fL$, by
\pr{1-hole congruence for image} \plat{$\sim^{1c}_{\fL\!,\bR}$} is a congruence for $\fT(\fL)$ on ${\bR}(\bV)$.
Since ${\bR}(\bV)=\bW$, which for \plat{$\sim^{1c}_{\fL\!,\bR}$} is the $\bU$ used in \df{weak congruence},
\plat{$\sim^{1c}_{\fL\!,\bR}$} is a congruence for $\fT(\fL)$.
So by \thm{valid correct}, $\fT$ is correct up to \plat{$\sim^{1c}_{\fL\!,\bR}$}.
\end{proof}

The languages $\pi$ and $\rm a\pi$ of \sect{asynpi} do not feature operators (or other constructs)
of infinite arity. Hence the congruence closure $\sim_\pi^{1c}$ or $\sim_{\rm a\pi}^{1c}$ of an
equivalence $\sim$ on $\pi$ or $\rm a\pi$ is always a congruence.
So by \cor{congruence closure of translation} Boudol's translation $\fT$ is correct up to an equivalence \plat{$\wbb_{\pi,{\bR}}^c$},
defined on the disjoint union of the domains $\T_\pi$ and $\T_{\rm a\pi}$ on which the two languages are
interpreted. This equivalence is contained in $\wbb$, and on the source domain $\T_\pi$ coincides
with $\cong^c$. By \thm{congruence closure of translation},  the restriction of \plat{$\wbb_{\pi,{\bR}}^c$}
to a subdomain $\bW\subseteq \T_{\rm a\pi}$ is the largest congruence for $\fT(\pi)$ on $\bW$ that is contained in $\wbb$.
As $\cong^c_{\rm a}$ is a congruence for all of $\rm a\pi$ on all of $\T_{\rm a\pi}$, and contained
in $\wbb$, it is certainly a congruence for $\fT(\pi)$ on $\bW$, and thus contained in \plat{$\wbb_{\pi,{\bR}}^c$}.
This inclusion turns out to be strict.
As an illustration of that, note that $\bar x z.\textbf{0} |  \bar x z.\textbf{0} \cong^c \bar xz.\bar x z \textbf{0}$.
(This follows since these processes are strong (early) bisimilar \cite{SW01book} and thus strong
full bisimilar by \cite[Def. 2.2.2]{SW01book}.)
Consequently, their translations must be related by \plat{$\wbb_{\pi,{\bR}}^c$}. So, for distinct $u,v,y,w,x,z \in\N$,
$$(u)(\bar{x}u | u(v).(\bar{v}z | \textbf{0})) \big| (u)(\bar{x}u | u(v).(\bar{v}z | \textbf{0})) \wbb_{\pi,{\bR}}^c
(y)(\bar{x}y | u(w).(\bar{w}z | (u)(\bar{x}u | u(v).(\bar{v}z | \textbf{0})) )). $$
Yet, these processes are not $\cong^c_{\rm a}$-equivalent, as can be seen by putting them in a
context $x(y).x(y).\bar r(s) | X$. In this context, only the left-hand side has a weak barb
$\Downarrow_{\bar r}$.

\section{Integrating language features through translations}\label{sec:integrating}

The results of the previous section show how valid translations are satisfactory for comparing the
expressiveness of languages. If there is a valid translation $\fT$ from $\fL$ to $\fL'$ up to $\sim$,
and (as usual) $\sim^{1c}_\fL$ is a congruence, then all truths that can be expressed in terms of
$\fL$ can be mimicked in $\fL'$. For the congruence classes of $\sim^{1c}_\fL$ translate bijectively
to congruence classes of an induced equivalence relation on the domain of $\fT(\fL)$ (within the
domain of $\fL'$), and all operations on those congruence classes that can be performed by contexts of
$\fL$ have a perfect counterpart in terms of contexts of $\fT(\fL)$. This state of affairs was
illustrated on Boudol's translation from a synchronous to an asynchronous $\pi$-calculus.

There is however one desirable property of translations between languages that has not yet been
achieved, namely to combine the powers of two languages into one unified language.
If both languages $\fL_1$ and $\fL_2$ have valid translations into a language $\fL'$,
then all that can be done with $\fL_1$ can be mimicked in a fragment of $\fL'$, and 
all that can be done with $\fL_2$ can be mimicked in another fragment of $\fL'$.
In order for these two fragments to combine, one would like to employ a single congruence relation
on $\fL'$ that specialises to congruence relations for $\fT_1(\fL_1)$ and $\fT_2(\fL_2)$,
which form the counterparts of relevant congruence relations for the source languages 
$\fL_1$ and $\fL_2$.

In terms of the translation $\fT$ from $\pi$ to $\rm a\pi$, the equivalence $\cong^c_{\rm a}$ on
$\T_{\rm a\pi}$ would be the right congruence relation to consider for $\rm a\pi$.
Ideally, this congruence would extend to an equivalence $\cong^c_{\pi,\rm a\pi}$ on the disjoint
union $\T_\pi \uplus \T_{a\pi}$, such that the restriction of $\cong^c_{\pi,\rm a\pi}$ to $\T_\pi$
is a congruence for $\pi$. Necessarily, this congruence on $\T_\pi$ would have to distinguish the
terms $\bar x z.\textbf{0} |  \bar x z.\textbf{0}$ and $\bar xz.\bar x z \textbf{0}$, since their
translations are distinguished by $\cong^c_{\rm a}$. One therefore expects $\cong^c_{\pi,\rm a\pi}$
on $\T_\pi$ to be strictly finer than $\cong^c$. Here it is important that the union
of $\T_\pi$ and $\T_{a\pi}$ on which this congruence is defined is required to be disjoint.
For if one considers $\T_{a\pi}$ as a subset of $\T_\pi$, then we obtain that the restriction of
$\cong^c_{\pi,\rm a\pi}$ to that subset (1) coincides with $\cong^c_{\rm a}$ and (2) is strictly
finer than $\cong^c$. This contradicts the fact that $\cong^c$ is strictly finer than $\cong^c_{\rm a}$.

In \sect{reflect} I will show that such a congruence $\cong^c_{\pi,\rm a\pi}$ indeed exists.
In fact, under a few very mild conditions this result holds generally, provided that the source
language $\fL$ is a closed-term language. The following example illustrated why that restriction
needs to be imposed.

\begin{example}{reflection}
Let $\fL$ be a language with three constants $\top$, \textbf{Yes} and \textbf{No}.
Its semantics is given by $\bV=\{+,-\}$, $\top^\bV =\textbf{Yes}^\bV=+$ and $\textbf{No}^\bV=-$.
Let $\fL'$ be a language with the same three constants and a unary operator \textbf{next}.
Its semantics is given by $\bV'=\{0,1,2\}$ with $\top^{\bV'}=2$, $\textbf{Yes}^{\bV'}=1$ and
$\textbf{No}^{\bV'}=0$, while $\textbf{next}^{\bV'}(n)=(n+1) \textsf{mod}\,3$.
Let $\sim$ be the semantic equivalence on $\bV' \cup \bV$ given by $- \sim 0 \not\sim 1 \sim 2 \sim +$.
Then the identity function $\cal I$ is a translation from $\fL$ into $\fL'$ that is valid up to $\sim$.
This is witnessed by the semantic translation ${\bR} :=\{(0,-),(1,+),(2,+)\} \subseteq {\sim}$.
The congruence closure of $\sim$ on $\bV'$ is the identity relation. This relation cannot be
extended to $\bV \cup \bV'$ in such a way that $\fT$ remains valid. For $1$ and $2$ need to be
different; yet to make $\fT$ valid both need to be related to $+$.
\end{example}

\section{A unique decomposition of terms}\label{sec:standard heads}

The results of Sections~\ref{sec:compositionality}--\ref{sec:preservation} apply only to languages that satisfy
two postulates, and to preorders $\asim$ that ``respect $\eqa$'' (defined in \sect{invariance}).
Below and in \sect{invariance} I formulate these postulates.

\begin{definition}{alpha}
\emph{$\alpha$-conversion} is the act of renaming all occurrences of a bound variable $X$ within the
scope of its binding into another variable, say $Y$, while avoiding capture of free variables. Here
one speaks of \emph{capture} when a free occurrence of $Y$ turns into a bound one.

Write $E \eqa F$ if expression $E$ can be converted into $F$ by multiple acts of $\alpha$-conversion.
\end{definition}
In languages where there are multiple types of bound variables, $\eqa$ allows conversion of all of them.
In a $\pi$-calculus with recursion, for instance, there could be bound process variables $X\in\V$ as well
as bound names $x\in\N$.
The last two conversions in the right column of \df{structural congruence} define $\alpha$-conversion for $\pi$-calculus names.

The following notation and observation is used below and in \sect{by construction}.

\begin{definition}{composition of substitutions}
Given two substitutions $\sigma,\xi:\V\rightharpoonup\IT_\fL$, their composition $\xi\bullet\sigma$
with $\dom(\xi\bullet\sigma)=\dom(\sigma)$ is given by $(\xi\bullet\sigma)(X)=\sigma(X)[\xi]$.
\end{definition}

\begin{observation}{composition of substitutions}
If $\dom(\sigma)=\fv(E)$ then $E[\sigma][\xi]\eqa E[\xi\bullet\sigma]$.
\end{observation}
An expression $E$ of the form $f(E_1,\ldots,E_n\hspace{-1pt})$ can be written as $H[\sigma]$, where $H$ is an
expression $f(X_1,\ldots,X_n\hspace{-1pt})$ and $\sigma:\{X_1,\ldots,X_n\}\rightarrow\IT_\fL$ is the substitution
with $\sigma(X_i)=E_i$ for $i=1,\ldots,n$. Here $H$ and $\sigma$ are completely determined by $E$,
except for the choice of the variables $X_1,\ldots,X_n$. The term $H$ is called a \emph{head} of $E$.
In this paper, for the proofs of \sect{compositionality}, I propose a unique decomposition of expressions $E$ into
$H$ and $\sigma$, by making an arbitrary choice for the variables $X_1,\ldots,X_n$. Moreover, I
extend this decomposition to all terms $E$ that are not variables. This requires a postulate that
says, in essence, that such a decomposition is always possible, and restricting attention to
languages satisfying this postulate.

\begin{definitionc}{prefix}{\cite{vG12}}
A term $E\mathbin\in\IT_\fL$ is a \emph{prefix} of a term $F\!$, written $E\mathbin\leq F\!$, if
$F\mathbin{\stackrel{\alpha}=}E[\sigma]$ for some substitution $\sigma$ with $\dom(\sigma)=\fv(E)$.
\end{definitionc}
Note that $E[\textit{id}_E]=E$, where $\textit{id}_E:\fv(E)\rightarrow\IT_\fL$ is the identity.
Moreover, if $\dom(\sigma)=\fv(E)$ then $E[\sigma][\xi]\eqa E[\xi\bullet\sigma]$, by \obs{composition of substitutions}.
It follows that $\leq$ is reflexive and transitive, and hence a preorder.
Write $\equiv$ for the kernel of $\leq$, i.e.\ $E\equiv F$ iff $E\leq F \wedge F\leq E$.
Note that $E\eqa F$ implies $E\equiv F$. If $E\equiv F$ then $E$ can be
converted into $F$ by means of an injective renaming of its variables.

\begin{definitionc}{head}{\cite{vG12}}
A term $H\in\IT_\fL$ is a \emph{head} if $H$ is not a single variable and $E\leq H$
implies that $E$ is a single variable or $E\equiv H$.
It is a \emph{head of} another term $F$ if it is a head, as well as a prefix of $F$.
\end{definitionc}

\begin{example}{head}
Let $c$ be a constant, $g$ a unary operator, $f$ a binary operator,
and $\mu X.E$ the recursion construct of CCS\@.
Then $f(X,Y)$ is a head of the term $f(c,g(c))$, and
$\mu X.f(Y,g(g(X)))$ is a head of $\mu X.f(g(c),g(g(X)))$.
See \cite{vG12} for further detail.
\end{example}

\begin{postulatec}{head}{\cite{vG12}}
Each expression $E$, if not a variable, has a head, which is unique up to $\equiv$.
\end{postulatec}
This is easy to show for each common type of system description language, and I am not
aware of any counterexamples.  However, while striving for maximal generality, I consider
languages with (recursion-like) constructs that are yet to be invented, and in view of
those, this principle has to be postulated rather than derived.
This means that in this section I consider only languages that satisfy this postulate.

Pick a representative from each $\equiv$-equivalence class of
heads, and call the chosen representatives \emph{standard heads}.
Likewise, pick a representative from each $\eqa$-class of terms,
and call a substitution $\sigma$ where $\sigma(X)$ is such a
representative for each $X\in\dom(\sigma)$ a \emph{standard substitution}.
Now each expression $E\notin\V$ can uniquely be written as $E\eqa
H[\sigma]$, with $H$ a standard head and $\sigma$ a standard substitution
with $\dom(\sigma)=\fv(E)$. I will refer to the pair $H,\sigma$ as the
\emph{standard decomposition} of $E$.

\section[Invariance of meaning under alpha-conversion]
        {Invariance of meaning under $\alpha$-conversion}\label{sec:invariance}

Write $\val \eqa_{\fL} \wal$, with $\val,\wal\in \bV$, iff there are terms $E,F\in\IT_\fL$ with $E\eqa F$, and a valuation
$\zeta:\V\rightarrow\bV$ such that $\denote{E}_\fL(\zeta)=\val$ and $\denote{F}_\fL(\zeta)=\wal$.
This relation is reflexive and symmetric.

In \cite{vG12} I limited attention to languages satisfying
\begin{equation}\label{alpha}\textit{
if $E\eqa F$ then $\denote{E}_\fL=\denote{F}_\fL$.}
\end{equation}
This postulate says that the meaning of an expression is invariant under $\alpha$-conversion.
It can be reformulated as the requirement that $\eqa_{\fL}$ is the identity relation.
This postulate is satisfied by all my intended applications, except for the important class
of closed-term languages.
Languages like CCS and the $\pi$-calculus can be regarded as falling
in this class (although it is also possible to declare the meaning of a term under a valuation to be
an $\eqa$-equivalence class of closed terms).
To bring this type of application within the scope of my theory, here
I weaken this postulate by requiring merely that $\eqa_{\hspace{-2pt}\fL}$ is an equivalence.%
\begin{postulate}{alpha}
$\eqa_\fL$ is an equivalence relation.
\end{postulate}
This postulate is needed for the results of Sections~\ref{sec:compositionality}--\ref{sec:preservation}.
I also need to restrict attention to preorders $\asim$ with $\mathord{\eqa_\fL} \subseteq \mathord{\asim}$.
When that holds I say that the preorder $\asim$ \emph{respects} $\eqa_\fL$.
If (\ref{alpha}) holds---which strengthens of \pos{alpha}---then \emph{any} preorder respects $\eqa_\fL$.

\section{Compositionality}\label{sec:compositionality}

An important property of translations, defined below, is \emph{compositionality}.
In this section show I that any valid translation up to a preorder $\asim$ can be modified into such
a translation that moreover is compositional, provided one restricts attention to languages
that satisfy Postulates~\ref{post:head} and~\ref{post:alpha}, and preorders $\asim$ that respect $\eqa$.

\begin{definition}{compositionality}
A translation $\fT$ from $\fL$ into $\fL'$ is \emph{compositional} if
\begin{enumerate}[(1)]
\item\label{comp1}
$\fT(E[\sigma])\eqa \fT(E)[\fT\circ\sigma]$ for each $E\in\IT_\fL$ and $\sigma:\fv(E)\rightarrow\IT_\fL$,
\item\label{comp2}
$E\eqa F$ implies $\fT(E) \eqa \fT(F)$ for all $E,F\in\IT_\fL$,
\item\label{comp3}
and moreover $\fT(X)=X$ for each $X\in\V$.
\end{enumerate}
\end{definition}
In case $E=f(t_1,\ldots,t_n)$ for certain $t_i\in\IT_\fL$ this amounts to
$\fT(f(t_1,\ldots,t_n)) \eqa  E_f(\fT(t_1),\ldots,\fT(t_n))$,
where $E_f:=\fT(f(X_1,\ldots,X_n))$ and $E_f(u_1,\ldots,u_n)$ denotes the result of the simultaneous
substitution in this expression of the terms $u_i\in\IT_{\fL'}$ for the free variables $X_i$,
for $i=1,\ldots,n$. The first requirement of \df{compositionality} is more general and covers
language constructs other than functions, such as recursion. Requiring equality rather than $\eqa$
is too demanding, as the following example illustrates.
\begin{example}{alpha compositionality}
Take a source language $\fL$ that features an unary replication operator $!$, as in the
$\pi$-calculus, and a target language $\fL'$ that instead has a recursion construct $\mu X.E$,
as in CCS; both languages have a constant $0$. A suitable translation $\fT$ satisfies
$\fT(!X_1) = \mu X. (X|X_1)$ and $\fT(0)\mathbin=0$.
Applying \df{compositionality} with $\sigma(X_1)=0$ gives $\fT(!0) \eqa \mu X. (X|0)$, whereas
applying it with $\sigma(X_1)=X$ gives $\fT(!X) \eqa \mu Y. (Y|X)$.
Here the bound variable $X$ needed to be renamed (into $Y$) to avoid capture of the free
variable $X$ that is substituted for $X_1$. Furthermore, applying \df{compositionality} on
$\fT(!X) \eqa \mu Y. (Y|X)$ with $\sigma(X)=0$ gives $\fT(!0) \eqa \mu Y. (Y|0)$. Since $X\neq Y$,
this shows that $\eqa$ cannot consistently be replaced by $=$.
\end{example}

\begin{lemma}{composition is compositional}
If $\fT_1:\IT_{\fL_1} \rightarrow \IT_{\fL_2}$ and $\fT_2:\IT_{\fL_2} \rightarrow \IT_{\fL_3}$ are
compositional translations, then so is their composition
$\fT_2\circ\fT_1:\IT_{\fL_1}\rightarrow\IT_{\fL_3}$, defined by
$\fT_2\circ\fT_1(E):=\fT_2(\fT_1(E))$ for all $E\in\fL_1$.
\end{lemma}
\begin{proof}
(1) $\fT_2(\fT_1(E[\sigma]))\eqa \fT_2(\fT_1(E)[\fT_1\circ\sigma]) \eqa
\fT_2(\fT_1(E))[\fT_2\circ\fT_1\circ\sigma])$ for each $\sigma:\V\rightharpoonup\IT_{\fL_1}$ and $E\in\IT_{\fL_1}$.
Here the derivation of the first $\eqa$ uses Property (\ref{comp2}) of \df{compositionality}---and
this is the reason for requiring that property.

(2) $E\eqa F$ implies $\fT_1(E) \eqa \fT_1(F)$ and hence $\fT_2(\fT_1(E)) \eqa \fT_2(\fT_1(F))$ for
all $E,F\in\IT_\fL$.

(3) $\fT_2(\fT_1(X)) = \fT_2(X)=X$ for each $X\in\V$.
\end{proof}

\subsection{Translations that are compositional by construction}\label{sec:by construction}

The following proposition shows that when verifying that a translation is compositional it suffices
to check requirement (\ref{comp1}) of \df{compositionality} for the case that the term $E$ is a
(standard) head only. In most applications (including the one of \sect{asynpi}) a translation is
defined inductively, in such a way that (\ref{comp1}) for $E$ a head, as well as (\ref{comp2}) and
(\ref{comp3}), hold by definition; in such cases compositionality follows.

\begin{proposition}{compositionality}
Let $\fL$ be a language satisfying \pos{head}.
Any translation $\fT$ from $\fL$ into $\fL'$ satisfying
\begin{center}
\begin{tabular}{ll}
$\fT(X) = X$ & for $X\in \V$, and\\
$\fT(E) \eqa \fT(H)[\fT\circ\sigma]$ & when $E\eqa H[\sigma]$ with $H,\sigma$ the standard
  decomposition of $E$.
\end{tabular}
\end{center}
is compositional.
\end{proposition}

\begin{proof}
By assumption, this translation satisfies Properties (\ref{comp2}) and (\ref{comp3}) of \df{compositionality}.
Next, I show that $\fT$ also satisfies Property (\ref{comp1}), using induction on $E/_{\!\!\stackrel{\alpha}=}$.
So let $E\mathbin\in\IT_\fL$ and $\xi\!:\fv(E)\rightarrow\IT_\fL$.
I have to show that $\fT(E[\xi])\eqa \fT(E)[\fT\circ\xi]$.
The case $E\in\V$ is trivial, so let $E\eqa H[\sigma]$, with $H$ a standard head and $\sigma$ a
standard substitution with $\dom(\sigma)=\fv(H)$.
For each free variable $X$ of $H$, $\sigma(X)$ is a proper subterm of $E$ up to $\eqa$,
so by the induction hypothesis $\fT(\sigma(X)[\xi])\eqa \fT(\sigma(X))[\fT\circ\xi]$.
Thus, for $X\in\fv(H)$,\\
$\begin{array}[t]{@{}l@{~}ll@{}}
(\fT\!\circ(\xi\bullet\sigma))(X)
&= \fT((\xi\bullet\sigma)(X))
& \mbox{by definition of functional composition $\circ$}
\\
&= \fT(\sigma(X)[\xi])
& \mbox{by definition of the operation $\bullet$ on substitutions}
\\
&\eqa \fT(\sigma(X))[\fT\circ\xi]
& \mbox{by induction, derived above}
\\
&= ((\fT\!\circ\xi)\bullet(\fT\!\circ\sigma))(X)
& \mbox{by definition of the operations $\circ$ and $\bullet$.}
\end{array}$\\
This shows that the substitutions $\fT\!\circ(\xi\bullet\sigma)$ and
$(\fT\!\circ\xi)\bullet(\fT\!\circ\sigma)$---both with domain $\fv(H)$---are equal up to $\alpha$-conversion,
from which it follows that $F[\fT\circ(\xi\bullet\sigma)] \eqa (F[\fT\circ\sigma])[\fT\circ\xi]$
for all terms $F\in\IT_{\fL'}$.
\\
$\begin{array}[b]{@{}l@{~}ll@{}}
\mbox{Hence}~ \fT(E[\xi])
&\eqa \fT(H[\sigma][\xi])
& \mbox{since $E\eqa H[\sigma]$, and thus $E[\xi]\eqa H[\sigma][\xi]$, using (\ref{comp2}) }
\\
&\eqa \fT(H[\xi\bullet\sigma])
& \mbox{by \obs{composition of substitutions}}
\\
&\eqa \fT(H[\upsilon])
& \mbox{for a standard substitution $\upsilon \eqa \xi\bullet\sigma$}
\\
&\eqa \fT(H)[\fT\circ\upsilon]
& \mbox{by assumption}
\\
&\eqa \fT(H)[\fT\circ(\xi\bullet\sigma)]
& \mbox{replacing $\upsilon$ by $\xi\bullet\sigma$, using (\ref{comp2})}
\\
&\eqa (\fT(H)[\fT\circ\sigma])[\fT\circ\xi]
& \mbox{derived above}
\\
&\eqa \fT(E)[\fT\circ\xi]
& \mbox{by assumption.}
\end{array}$
\end{proof}

\subsection{The denotation of a substitution}

To obtain the compositionality of valid translations, I need the concept of the \emph{denotation} of
a substitution as a transformation of valuations.
The semantic mapping $\denote{\ \ }_{\fL}:\IT_{\fL} \rightarrow ((\V\!\rightarrow\bV)\rightarrow\bV)$
extends to substitutions $\sigma$ by $\denote{\sigma}_\fL(\rho)(X):=\denote{\sigma(X)}_\fL(\rho)$
for all $X\mathbin\in\V$ and $\rho:\V\!\!\rightarrow\bV$---here $\sigma$ is extended to a total function by
$\sigma(Y):=Y$ for all $Y\not\in\dom(\sigma)$. Thus $\denote{\sigma}_\fL$ is of type
$(\V\rightarrow\bV)\rightarrow(\V\rightarrow\bV)$, i.e.\ a map from valuations to valuations.
The inductive nature of the semantic mapping $\denote{\ \ }_\fL$ ensures that
for each expression $E\in\IT_\fL$, substitution $\sigma:\V\rightharpoonup\IT_\fL$ and
valuation $\rho:\V\rightarrow\bV$ there exists a term $F$ with $E\eqa F$ such that
$
\denote{E[\sigma]}_\fL(\rho) = \denote{F[\sigma]}_\fL(\rho) = \denote{F}_\fL(\denote{\sigma}_\fL(\rho))
$,
and hence
\begin{equation}\label{inductive meaning}
\denote{E[\sigma]}_\fL(\rho) \eqa_{\fL} \denote{E}_\fL(\denote{\sigma}_\fL(\rho)).
\end{equation}
In case $E$ is $f(X_1,\ldots,X_n)$ this amounts to
$\denote{f(E_1,\ldots,E_n)}_\fL(\rho) = f^\bV(\denote{E_1}_\fL(\rho),\ldots,\denote{E_n}_\fL(\rho))$,
but the above is more general and anticipates language constructs other than functions, such as recursion.

\subsection[Closing a semantic translation under alpha-conversion]
           {Closing a semantic translation under $\alpha$-conversion}

The following lemma (assuming \pos{alpha}) says that $\eqa_\fL$ is a congruence.

\begin{lemma}{alpha congruence}
Let $E\in\IT_\fL$ and $\nu,\rho:\V\rightarrow\bV$.
If $\nu \eqa_\fL \rho$ then $\denote{E}_\fL(\nu) \eqa_\fL \denote{E}_\fL(\rho)$.
\end{lemma}
\begin{proof}
Suppose $\nu \eqa_\fL \rho$. Then, for each $X\in\V$, $\nu(X) \eqa_\fL \rho(X)$, so there are terms
$E_X,F_X\in\IT_\fL$ with $E_X\eqa F_X$ and a valuation $\zeta_X:\V\rightarrow\bV$ such that
$\denote{E_X}_\fL(\zeta_X)=\nu(X)$ and $\denote{F_X}_\fL(\zeta_X)=\rho(X)$.
By renaming of variables one can assure that $\fv(E_X) \cap \fv(E_Y) = \emptyset$ for any different
$X,Y\in\fv(E)$. Here I assume that the set $\V$ of variables is sufficiently large.
Note that $\fv(F_X)=\fv(E_X)$ for all $X\in\fv(E)$.
Let $\zeta:\V\rightarrow\bV$ be a valuation satisfying $\zeta(Z)=\zeta_X(Z)$ for any $Z\in\fv(E_X)$
with $X\in\fv(E)$. 
Define the substitutions
$\sigma,\xi:\fv(E)\rightarrow\IT_\fL$ by $\sigma(X)=E_X$ and $\xi(X)=F_X$ for all $X\in\fv(E)$.
Then $\denote{\sigma(X)}_\fL(\zeta)=\nu(X)$ and $\denote{\xi(X)}_\fL(\zeta)=\rho(X)$ for all $X\in\fv(E)$.
Hence, using (\ref{inductive meaning}),\vspace{1pt}
$\denote{E}_\fL(\nu) = \denote{E}_\fL(\denote{\sigma}_\fL(\zeta)) \eqa_\fL \denote{E[\sigma]}_\fL(\zeta)$
and $\denote{E}_\fL(\rho) = \denote{E}_\fL(\denote{\xi}_\fL(\zeta)) \eqa_\fL \denote{E[\xi]}_\fL(\zeta)$.\vspace{1pt}
As $\sigma(X)\eqa\xi(X)$ for all $X\mathbin\in\V$ one has $E[\sigma]\eqa E[\xi]$, and thus
$\denote{E}_\fL(\nu) \eqa_\fL \denote{E[\sigma]}_\fL(\zeta) \eqa_\fL \denote{E[\xi]}_\fL(\zeta)
\eqa_\fL \denote{E}_\fL(\rho)$. Hence $\denote{E}_\fL(\nu) \eqa_\fL \denote{E}_\fL(\rho)$ with \pos{alpha}.
\end{proof}
Given a relation $\mathord{\bR}\subseteq\bV'\times\bV$, define $\bR^\alpha$ by
$\wal \bR \val$ iff $\exists \wal',\val'.\; \wal \eqa_{\fL'} \wal' \bR \val' \eqa_\fL \val$.

\begin{lemma}{correct up to alpha}
If a translation $\fT$ between languages $\fL$ and $\fL'$ that satisfy \pos{alpha} is correct
w.r.t.\ a semantic translation $\bR$, then it is also correct w.r.t.\ $\bR^\alpha$.
\end{lemma}

\begin{proof}
Let $\mathord{\bR}\subseteq\bV'\times\bV$ be a semantic translation, and $\fT$ a translation that is
correct w.r.t.\ $\bR$. Let $E\in\IT_\fL$, $\eta:\V\rightarrow\bV'$ and $\rho:\V\rightarrow\bV$, with
$\eta \bR^\alpha \rho$. Then there must be valuations $\theta:\V\rightarrow\bV'$ and $\nu:\V\rightarrow\bV$
with $\eta \eqa_{\fL'} \theta \bR \nu \eqa_\fL \rho$.
Since $\fT$ is correct w.r.t.\ $\bR$ one has $\denote{\fT(\E)}_{\fL'}(\theta) \bR \denote{\E}_\fL(\nu)$.
By \lem{alpha congruence}
$\denote{\fT(\E)}_{\fL'}(\eta) \eqa_{\fL'} \denote{\fT(\E)}_{\fL'}(\theta) \bR \denote{\E}_\fL(\nu) \eqa_\fL \denote{\E}_\fL(\rho)$
and thus $\denote{\fT(\E)}_{\fL'}(\eta) \bR^\alpha \denote{\E}_\fL(\rho)$.
\end{proof}

\subsection{Only the effect on standard heads matters}\label{only heads matter}

\begin{proposition}{heads suffice}
Let $\fL$ and $\fL'$ be languages that satisfy Postulates~\ref{post:head} and~\ref{post:alpha},
and $\mathord{\bR}\subseteq\bV'\times\bV$ a semantic translation.
A compositional translation $\fT:\IT_\fL\rightarrow\IT_{\fL'}$ is correct w.r.t.\ $\bR^\alpha$
iff $\denote{\fT(H)}_{\fL'}(\eta) \bR^\alpha \denote{H}_\fL(\rho)$ for all standard heads $H\in \IT_\fL$
and all valuations $\eta:\V\rightarrow\bV'$ and $\rho:\V\rightarrow\bV$ with $\eta\bR^\alpha\rho$.
\end{proposition}
\begin{proof}
``Only if'' follows immediately from \df{correct R}. So
assume $\denote{\fT(H)}_{\fL'}(\eta) \bR^\alpha \denote{H}_\fL(\rho)$ for all standard heads $H\in \IT_\fL$
and all $\eta:\V\rightarrow\bV'$ and $\rho:\V\rightarrow\bV$ with $\eta\bR^\alpha\rho$.
I have to show that $\fT$ is correct w.r.t.\ $\bR^\alpha$, i.e.\
that $\denote{\fT(\E)}_{\fL'}(\eta) \bR^\alpha \denote{\E}_\fL(\rho)$ for all
terms $\E\in \IT_\fL$ and all $\eta:\V\rightarrow\bV'$ and
$\rho:\V\rightarrow\bV$ with $\eta\bR^\alpha\rho$.
Let $\eta$ and $\rho$ be such valuations.
I proceed with structural induction on $E$, up to $\eqa$.
When handling a term $E\eqa H[\sigma]$, $\sigma(X)$ is, up to $\eqa$, a proper subterm of $E$ for
each free variable $X$ of $H$.
So by induction $\denote{\fT(\sigma(X))}_{\fL'}(\eta) \bR^\alpha \denote{\sigma(X)}_\fL(\rho)$.
The valuation $\denote{\sigma}_\fL(\rho)$ is defined such that
$\denote{\sigma}_\fL(\rho)(X)=\denote{\sigma(X)}_\fL(\rho)$ for each $X\in\V$.
Likewise, $\denote{\fT\circ\sigma}_{\fL'}(\eta)(X)=\denote{\fT(\sigma(X))}_{\fL'}(\eta)$ for each $X\in\V$.
Hence $\denote{\fT\circ\sigma}_{\fL'}(\eta) \bR^\alpha \denote{\sigma}_\fL(\rho)$.\hfill(*)
\begin{list}{$\bullet$}{\leftmargin 10pt}
\item $\denote{\fT(X)}_{\fL'}(\eta) = \denote{X}_{\fL'}(\eta)
\begin{array}[t]{@{~}ll}
=\eta(X) & \mbox{by definitions of $\fT$ and $\denote{\ \ }_{\fL'}$} \\
\bR^\alpha\rho(X) & \mbox{since $\eta\bR^\alpha\rho$} \\
=\denote{X}_{\fL}(\rho) & \mbox{by definition of $\denote{\ \ }_{\fL}$.}
\end{array}$
\item $\denote{\fT(E)}_{\fL'}(\eta)
\begin{array}[t]{@{~}l@{~}ll@{}}
\eqa_{\fL'} & \denote{\fT(H)[\fT\circ\sigma]}_{\fL'}(\eta) & \mbox{by the compositionality of $\fT$, since $E\eqa H[\sigma]$} \\
\eqa_{\fL'} & \denote{\fT(H)}_{\fL'}(\denote{\fT\circ\sigma}_{\fL'}(\eta)) & \mbox{by (\ref{inductive meaning})} \\
\bR^\alpha & \denote{H}_\fL(\denote{\sigma}_\fL(\rho)) & \mbox{by assumption, using (*) above\hspace{85pt}} \\
\eqa_\fL & \denote{H[\sigma]}_\fL(\rho)  & \mbox{by (\ref{inductive meaning})}\\
\eqa_\fL  & \denote{E}_\fL(\rho) & \mbox{since $E\eqa H[\sigma]$.}
\hfill\mbox{\qed}
\end{array}$
\end{list}
\end{proof}

\subsection{Correct and valid transitions can be made compositional}\label{sec:compo}

\begin{theorem}{compositionality}
\hspace{-1pt}Let $\fL$ and $\fL'$ be languages satisfying Postulates~\ref{post:head} and~\ref{post:alpha}.
If any correct translation from $\fL\!$ into $\fL'\!$ w.r.t.\ $\bR^\alpha$ exists,
then there is a compositional translation from $\fL$ into $\fL'\!$ that is correct w.r.t.\ $\bR^\alpha\!\!$.
\end{theorem}

\begin{proof}
Given a translation $\fT_0$ that is correct w.r.t.\ $\bR^\alpha$, define the translation $\fT$ inductively by
\begin{center}
\begin{tabular}{ll}
$\fT(X) := X$ & for $X\in \V$\\
$\fT(E) := \fT_0(H)[\fT\circ\sigma]$ & when $E\eqa H[\sigma]$ with $H,\sigma$ the standard
  decomposition of $E$.
\end{tabular}
\end{center}
By filling in $H[\textit{id}_H]$ for $E$, with $H$ a standard head and $\textit{id}_H$ the identity
substitution on $\fv(H)$, I obtain $\fT(H)=\fT_0(H)$ for each standard head $H$.
Hence, by \pr{compositionality}, $\fT$ is compositional.
Moreover, since $\fT_0$ is correct w.r.t.\ $\bR^\alpha$, for all standard heads $H\in \IT_\fL$
and all $\eta:\V\rightarrow\bV'$ and $\rho:\V\rightarrow\bV$ with $\eta\bR^\alpha\rho$
one has $\denote{\fT(H)}_{\fL'}(\eta) = \denote{\fT_0(H)}_{\fL'}(\eta) \bR^\alpha \denote{H}_\fL(\rho)$.
Thus, by \pr{heads suffice}, $\fT$ is correct w.r.t.\ $\bR^\alpha$.
\end{proof}

\begin{corollary}{compositionality}
Let $\fL$ and $\fL'$ be languages that satisfy Postulates~\ref{post:head} and~\ref{post:alpha},
and $\asim$ a preorder that respects $\eqa_\fL$ and $\eqa_{\fL'}$.
If any valid (or correct) translation from $\fL$ into $\fL'$ up to $\asim$ exists,
then there exists a compositional translation that is valid (or correct) up to $\asim$.
\end{corollary}
\begin{proof}
Let $\fT: \IT_\fL \rightarrow \IT_{\fL'}$ be valid up to $\asim$. Then
$\fT$ is correct w.r.t.\ some semantic translation $\mathord{\bR}\subseteq\bV'\times\bV$
with $\mathord{\bR} \subseteq \mathord{\asim}$.
By \lem{correct up to alpha} $\fT$ is correct w.r.t.\ $\bR^\alpha$.
By \thm{compositionality} there exists a compositional translation $\fT': \IT_\fL \rightarrow \IT_{\fL'}$
that is correct w.r.t.\ $\bR^\alpha$.
Since $\asim$ respects $\eqa_\fL$ and $\eqa_{\fL'}$, one has
$\mathord{\eqa_\fL},\mathord{\eqa_{\fL'}},\mathord{\bR} \subseteq \mathord{\asim}$, and thus
$\mathord{\bR^\alpha} \subseteq \mathord{\asim}$.
The statement about correct translations up to $\asim$ follows in the same way,
or as a corollary by use of \thm{valid correct}.
\end{proof}
Hence, for the purpose of comparing the expressive power of languages,
valid translations between them can be assumed to be compositional.
For correct translations this was already established in \cite{vG12}, but assuming (\ref{alpha}),
a stronger version of \pos{alpha}.

\section{Translations from closed-term languages reflect target congruences}\label{sec:reflect}

I can now establish the theorem promised in \sect{integrating}.
In view of \cor{compositionality}, no great sacrifices are made by assuming that the translation
$\fT$ is compositional. Additional ``mild conditions'' needed here are
\pos{alpha} for $\fL'$ and $\approx$ respecting $\eqa_{\fL'}$.

\begin{theorem}{pulling back congruences}
Let $\fL$ be a closed-term language and $\fL'$ a language that satisfies \pos{alpha}.
Let $\fT$ be a compositional translation from $\fL$ into $\fL'$ that is valid up to $\sim$.
Let $\approx$ be any congruence for $\fL'$ containing $\eqa_{\fL'}$ and contained in $\sim$.
  Then $\fT$ is correct up to an equivalence $\approx_{\fT}$ on $\bV \cup \bV'$, contained in
  $\sim$, that on $\bV'$ coincides with $\approx$. 
\end{theorem}

\begin{proof}
  If $\bV=\emptyset$, the statement is trivial. Otherwise, pick a valuation $\zeta:\V\rightarrow\bV$.
  By assumption, $\fT$ is correct w.r.t.\ a semantic translation ${\bR} \subseteq {\sim}$.
  Now pick a valuation $\theta:\V\rightarrow \bV'$ with $\theta \bR \zeta$.
  Let $\equiv_{\fT}$ be the smallest equivalence relation on $\bV \cup \bV'$ such that
  $\denote{\fT(\p)}_{\fL'}(\theta) \equiv_{\fT} \denote{\p}_\fL$ for all $\p\in\T_\fL$.
  Since $\fT$ is correct w.r.t.\ $\bR$, one has
  $\denote{\fT(\p)}_{\fL'}(\theta) \bR \denote{\p}_\fL(\zeta) = \denote{\p}_\fL$ for all $\p\in\T_\fL$.
  Hence ${\equiv_{\fT}} \subseteq {\sim}$.

Define $\approx_{\fT}$ on $\bV \cup \bV'$ by \plat{$\val_1 \approx_{\fT} \val_2$} iff
\plat{$\val_1 \equiv_{\fT} \wal_1 \approx \wal_2 \equiv_{\fT} \val_2$} for some $\wal_1,\wal_2\in{\bV'}$.
Then ${\approx_{\fT}} \subseteq {\sim}$.
Since $\fL$ is a closed-term language, for each $\val \in \bV$ there is exactly one $\p\in \T_\fL$
with $\denote{\p}_\fL=\val$, namely $\p=\val$.
Hence, each $\equiv_{\fT}$-equivalence class on $\bV \cup \bV'$ contains exactly one element of $\bV'$.
It follows that $\approx_{\fT}$ is transitive, and hence an equivalence relation. Moreover, on $\bV'$
it coincides with $\approx$. It remains to show that $\fT$ is valid up to $\approx_{\fT}$.
For then \thm{valid correct} implies that $\fT$ is correct up to $\approx_{\fT}$.

Let $\bR_\fT := \{(\wal,\val) \mid \val \in\T_\fL \wedge \wal \eqa_{\fL'} \denote{\fT(\val)}_{\fL'}(\theta)\}$.
Then $\bR_\fT$ is a semantic translation with ${\bR_\fT} \subseteq {\approx_\fT}$,
using that ${\eqa_{\fL'}} \subseteq {\approx}$.
To show validity of $\fT$ up to $\approx_{\fT}$ it suffices to establish that $\fT$ is correct w.r.t.\ $\bR_\fT$.
So let $E\mathbin\in\IT_\fL$ and let $\eta\!:\V\!\rightarrow\bV'$ and $\rho\!:\V\! \rightarrow \bV$
be valuations with $\eta \bR_\fT \rho$.
I have to show that $\denote{\fT(\E)}_{\fL'}(\eta) \bR_\fT \denote{\E}_\fL(\rho)$.
Since $\fL$ is a closed-term language, $\rho$ also is a closed substitution,
and $\denote{\E}_\fL(\rho) := E[\rho] =: \denote{E[\rho]}_\fL$.\vspace{1pt}
Let $\fT\circ\rho:\V\rightarrow\T_{\fL'}$ be the substitution with $(\fT\circ\rho)(X)=\fT(\rho(X))$ for all $X$.
Then---filling in $(\eta(X),\rho(X))$ for $(\wal,\val)$ in the definition of $\bR_\fT$---%
$\eta(X)\eqa_{\fL'}\denote{(\fT\circ\rho)(X)}_{\fL'}(\theta)$ for all $X\in\V$, i.e.,
$\eta\eqa_{\fL'}\denote{\fT\circ\rho}_{\fL'}(\theta)$.
So, using \lem{alpha congruence}, (\ref{inductive meaning}) and compositionality,
\[\denote{\fT(\E)}_{\fL'}(\eta)  \eqa_{\fL'} \denote{\fT(\E)}_{\fL'}(\denote{\fT\circ\rho}_{\fL'}(\theta)) \eqa_{\fL'}
\denote{\fT(\E)[\fT\circ\sigma])}_{\fL'}(\theta) \eqa_{\fL'} {}\vspace{-8pt}\]
\hfill ~~~~$\denote{\fT(\E[\sigma])}_{\fL'}(\theta)\bR_\fT \denote{\E[\sigma]}_\fL = \denote{\E}_\fL(\rho)$.\\[1ex]
It follows that $\denote{\fT(\E)}_{\fL'}(\eta) \bR_\fT \denote{\E}_\fL(\rho)$ using \pos{alpha}.
\end{proof}
Since by \thm{pulling back congruences} $\fT$ is correct up to $\approx_\fT$, by \pr{congruence}
$\approx_\fT$ is a congruence for $\fL$, contained in $\sim$.\linebreak[3] Consequently, on $\bV$, $\approx_\fT$
is contained in $\sim_\fL^c$, the coarsest 1-hole congruence for $\fL$ contained in $\sim$.
As remarked in \sect{integrating}, for the example of \sect{asynpi} this inclusion is strict.

\section[Comparison with the definition of validity from {[vG12]}]
        {Comparison with the definition of validity from \cite{vG12}}\label{sec:respects}

An earlier definition of validity occurs in \cite{vG12}.
A shortcoming of that notion was that it only applied to languages for which all values in the
domain of interpretation are denotable by closed terms.
Here I show that the current notion of validity, which does not suffer from this limitation,
generalises the one of \cite{vG12}.\linebreak[3]\vspace{-13.6pt}

\noindent
Let $\fL$ and $\fL'$ be languages with
$\denote{\ \ }_{\fL}:\IT_{\fL} \rightarrow ((\V\rightarrow\bV)\rightarrow\bV)$
and
$\denote{\ \ }_{\fL'}:\IT_{\fL'} \rightarrow ((\V\rightarrow\bV')\rightarrow\bV')$.
The property of a language that all values in the domain of interpretation are denotable by closed
terms can be stated as\vspace{-1ex}
\begin{equation}\label{denotable}
\forall v\in\bV.~\exists \p\in\T_\fL.~\denote{\p}_\fL=v\;.
\end{equation}

\begin{definitionc}{respects}{\cite{vG12}}
A translation $\fT$ from $\fL$ into $\fL'$ \emph{respects} $\asim$ if (\ref{related}) holds
and $\denote{\fT(\p)}_{\fL'}(\eta) \asim \denote{\p}_\fL$ for all closed $\fL$-expressions $\p\in\T_\fL$
and all valuations $\eta:\V\rightarrow\bU$, with $\bU:=\{\val\in\bV'\mid \exists \val\in \bV.~\val'\asim \val\}$.
\end{definitionc}
In case $\bV=\emptyset$, then $\bU=\emptyset$, so by lack of any
$\eta:\V\rightarrow\bU$ each translation $\fT$ from $\fL$ into $\fL'$ respects $\asim$.\linebreak[3]
When $\bV\neq\emptyset$ the valuation $\eta$ is needed because $\fT(\p)$ need not be a closed term,
even if $\p$ is.

\begin{definition}{valid-vG12}
In \cite{vG12}, a translation from $\fL$ into $\fL'$ is called \emph{valid up to} $\asim$ if
it is compositional and respects $\asim$, while $\fL$ satisfies (\ref{denotable}).
\end{definition}
Whereas (\ref{denotable}) is an unwanted limitation,
Example 2 in \cite{vG12} shows that simply skipping requirement (\ref{denotable}) in \df{valid-vG12}
would yield a criterion that is too weak. The solution proposed by the current paper
is to skip (\ref{denotable}), while simultaneously strengthening or replacing the requirement of
respecting $\asim$ by the one of
being valid up to $\asim$ (\sect{validity}). The following result state that, for languages $\fL$
and $\fL'$ that satisfy \pos{alpha}, and preorders $\asim$ that respects $\eqa_\fL$ and $\eqa_{\fL'}$,
any translation that is valid up to $\asim$ according to \cite{vG12} is certainly valid according to
\df{valid} of the current paper.

\begin{theorem}{vG12-valid}
Let $\fL$ and $\fL'$ satisfy \pos{alpha}, let $\fL$ satisfy (\ref{denotable}), and
let $\fT$ be a compositional translation from $\fL$ into $\fL'$ that respects $\asim$,
where $\asim$ respects $\eqa_\fL$ and $\eqa_{\fL'}$. Then $\fT$ is valid up to $\asim$.
\end{theorem}
\begin{proof}
In case $\bV=\emptyset$, there is a unique semantic transition \mbox{${\bR} \subseteq\bV'\times\bV$},
namely $\emptyset \subseteq {\asim}$.
Since there are no valuations $\rho:\V\rightarrow\bV$, $\fT$ is correct w.r.t.\ $\bR$, and thus valid up to $\asim$.

Otherwise, let $\theta:\V\rightarrow\bU:=\{\val\in\bV'\mid \exists \val\in \bV.~\val'\asim \val\}$---it exists by
(\ref{related})---and define the semantic translation ${\bR}\subseteq\bV'\times\bV$ by
${\bR} := \{(\denote{\fT(\p)}_{\fL'}(\theta),\denote{\p}_\fL) \mid \p\in\T_\fL\}$.
Then ${\bR} \subseteq {\asim}$, since $\fT$ respects $\asim$.
Since $\asim$ respects $\eqa_\fL$ and $\eqa_{\fL'}$, also ${\bR^\alpha} \subseteq {\asim}$.
Hence it suffices to show that $\fT$ is correct w.r.t.\ $\bR^\alpha$.
So let $\E\in \IT_\fL$, $\rho:\V\rightarrow\bV$ and $\eta:\V\rightarrow\bV'$ with $\eta\bR^\alpha \rho$.
I need to show that $\denote{\fT(\E)}_{\fL'}(\eta) \bR^\alpha \denote{\E}_\fL(\rho)$.

Pick $\bar\rho\!:\!\V\!\rightarrow\bV$ and $\bar\eta\!:\!\V\!\rightarrow\bV'$ such that
$\eta \eqa_{\fL'} \bar\eta \bR \bar\rho \eqa_\fL \rho$.
Then there is a closed substitution $\sigma:\V\rightarrow\T_\fL$ such that
$\bar\rho(X)=\denote{\sigma(X)}_\fL$ and $\bar\eta(X)=\denote{\fT(\sigma(X))}_{\fL'}(\theta)$ for
all $X\in X$, i.e., $\bar\rho=\denote{\sigma}_\fL$ and $\bar\eta=\denote{\fT\circ\sigma}_{\fL'}(\theta)$.
Hence,
$\denote{\fT(E)}_{\fL'}(\eta)
\begin{array}[t]{@{~}l@{~}ll@{}}
\eqa_{\fL'} & \denote{\fT(E)}_{\fL'}(\bar\eta) & \mbox{by \lem{alpha congruence}}\\
 =  & \denote{\fT(E)}_{\fL'}(\denote{\fT\circ\sigma}_{\fL'}(\theta)) & \mbox{derived above}\\
\eqa_{\fL'} & \denote{\fT(E)[\fT\circ\sigma]}_{\fL'}(\theta) & \mbox{by (\ref{inductive meaning})}\\
\eqa_{\fL'} & \denote{\fT(E[\sigma])}_{\fL'}(\theta) & \mbox{by compositionality of $\fT$} \\
\bR & \denote{E[\sigma]}_{\fL} & \mbox{by definition of $\bR$} \\
\eqa_\fL & \denote{E}_{\fL}(\denote{\sigma}_\fL) & \mbox{by (\ref{inductive meaning})} \\
 = & \denote{E}_{\fL}(\bar\rho) & \mbox{derived above}\\
 \eqa_\fL & \denote{E}_{\fL}(\rho) & \mbox{by \lem{alpha congruence}.}
\hfill\mbox{\qed}
\end{array}$
\end{proof}

\section{The case where all semantic values are denotable by closed terms}\label{sec:fvr}

Here I show that the notion of validity from the current paper agrees with the one of \cite{vG12}
when applied to languages for which all values are denotable by closed terms.

Usually one employs translations $\fT$ with the property that for any
$E\in\IT_\fL$ any free variable of $\fT(E)$ is also a free variable of $E$---I call these
\emph{free-variable respecting translations}, or \emph{fvr-translations} \cite{vG12}.\linebreak[3]
If $\fT_0$ in the proof of \thm{compositionality} is an fvr-translation, then so is $\fT$.
Hence, replaying the proof of \cor{compositionality}, one obtains:

\begin{observation}{compositionality}
Let $\fL$ and $\fL'$ be languages that satisfy Postulates~\ref{post:head} and~\ref{post:alpha},
and $\asim$ a preorder that respects $\eqa_\fL$ and $\eqa_{\fL'}$.
If any valid (or correct) fvr-translation from $\fL$ into $\fL'$ up to $\asim$ exists,
then there exists a compositional fvr-translation that is valid (or correct) up to $\asim$.
\end{observation}

An fvr-translation $\fT$ from $\fL$ into $\fL'$ respects $\asim$ iff either $\bV=\emptyset$, or
(\ref{related}) holds and $\denote{\fT(\p)}_{\fL'} \asim \denote{\p}_\fL$ for all closed
$\fL$-expressions $\p\in\T_\fL$.

\begin{observation}{fvr respects}
Any fvr-translation that is valid up to $\asim$ respects $\asim$.
\end{observation}

\begin{lemma}{fvr}
Let $\fL'$ satisfy (\ref{denotable}) and $\bV\neq\emptyset$.
If there is a translation from $\fL$ into $\fL'$ that is valid up to a preorder $\asim$,
then there is a fvr-translation from $\fL$ into $\fL'$ that is valid up to $\asim$.
\end{lemma}
\begin{proof}
Let $\fT$ be a translation from $\fL$ into $\fL'$ that is valid up to $\asim$.
Let ${\bR} \subseteq {\asim}$ be a semantic translation such that $\fT$ is correct w.r.t.\ $\bR$.
Let $\p\in\T_{\fL'}$ be a closed $\fL'$-term with $\denote{\p}_{\fL'}\bR \val$ for some
$\val\in\bV$---such a $\p$ exists by (\ref{denotable}).\linebreak[3]
Let the translation $\fT'$ be obtained from $\fT$ by defining $\fT'(E):=\fT(E)[\sigma_E]$ for
all $E\in\IT_\fL$---here $\sigma_E$ is the substitution that (only) substitutes $\p$ for all variables
$X$ that occur in $\fT(E)$ but not in $E$. Then $\fT'$ is a fvr-translation.
It remains to show that $\fT'$ is correct w.r.t.\ $\bR$.

\mbox{So let $\E\mathbin\in \IT_\fL$, $\eta\!:\!\V\!\mathbin\rightarrow\bV'$ and
$\rho\!:\!\V\!\mathbin\rightarrow\bV$ with $\eta\mathbin{\bR}\rho$.}
I aim to show that $\denote{\!\fT'\hspace{-.7pt}(\E)}_{\fL'}(\eta) \bR \denote{\E}_\fL(\rho)$.
Let $\bar\eta:\V\rightarrow\bV'$ and $\bar\rho:\V\rightarrow\bV$ be substitutions with
$\bar\eta(X)=\eta(X)$ and $\bar\rho(X)=\rho(X)$ for all $X\in\fv(E)$, such that
$\bar\eta(Y)=\denote{\p}_{\fL'}$ and $\bar\rho(Y)=\val$ for all $Y\in\fv(\fT(E))\setminus\fv(E)$.
Then $\bar\eta\bR\bar\rho$. Consequently,\\[1ex]
\mbox{}\hfill $\denote{\fT'(\E)}_{\fL'}(\eta) = \denote{\fT(\E)[\sigma_E]}_{\fL'}(\eta) =
\denote{\fT(\E)}_{\fL'}(\bar\eta) \bR \denote{\E}_\fL(\bar\rho) = \denote{\E}_\fL(\rho)$.
\end{proof}

\noindent
The following result states that in case both $\fL$ and $\fL'$ satisfy
(\ref{denotable}), as well as Postulates~\ref{post:head} and~\ref{post:alpha}, and $\asim$ respects
$\eqa_\fL$ and $\eqa_{\fL'}$, the validity-based notion of expressiveness from \cite{vG12} coincides
with the one here.

\begin{corollary}{respect}
Let $\fL$ and $\fL'$ be languages satisfying Postulates~\ref{post:head} and~\ref{post:alpha}, as
well as (\ref{denotable}), and $\asim$ a preorder respecting $\eqa_\fL$ and $\eqa_{\fL'}$.
There exists a valid translation from $\fL$ into $\fL'$ up to $\asim$ iff there exists a
compositional translation from $\fL$ into $\fL'$ that respects $\asim$.
\end{corollary}

\begin{proof}
Suppose a valid translation up to $\asim$ exists.
In case $\bV=\emptyset$, 
by \cor{compositionality} there exists a compositional translation $\fT$ that is valid up to $\sim$,
and by the remark following \df{respects} $\fT$ respects $\asim$.\\
\mbox{}\hspace{\parindent}%
In case $\bV\neq\emptyset$, by \lem{fvr} there exists an fvr-translation that is valid up to $\asim$.
By \obs{compositionality}, there exists a compositional fvr-translation $\fT$ that is valid up to $\sim$.
By \obs{fvr respects}, $\fT$ respects $\asim$.

The converse direction follows from \thm{vG12-valid}.
\end{proof}

\section{A potential generalisation of the concept of a valid translation}\label{sec:preservation}

In this section I consider a potentially more liberal concept of a valid transition up to a preorder $\asim$,
namely a translation that is compositional and \emph{preserves} $\asim$, as in \df{preservation} below.
Like correctness w.r.t.\ a semantic translation (\df{correct R}) and consequently also correctness
and validity up to $\asim$ (Definitions~\ref{df:correct}, \ref{df:valid} and ~\ref{df:correct translation}),
the requirement of \df{preservation} is that the translation preserves the meaning of
expressions: the meaning of the translation of an expression $E$ should be
semantically equivalent to the meaning of $E$\linebreak[1]---see \cite{vG12} for an elaboration.
In fact, this should hold under any valuation of the variables occurring in $E$.
The difference between Definitions~\ref{df:correct translation} and \ref{df:preservation}
is that the former is based on universal quantification of all matching valuations in the target
language, whereas the latter associates with any valuation in the source language a single
matching valuation in the target language. This requires
a translation $\fT:\IT_\fL\rightarrow\IT_{\fL'}$ to have a \emph{semantic counterpart}
$\bT:\bV\rightarrow\bV'$ that maps the possible meanings of $\fL$-expressions into the
possible meanings of $\fL'$-expressions.

\begin{definition}{preservation}
A translation $\fT$ from $\fL$ into $\fL'$ \emph{preserves $\asim$} iff
there exists a mapping $\bT:\bV \rightarrow\bV'$ such that
$\bT(\val)\asim \val$ for all $\val\mathbin\in\bV$ and
$\denote{\fT(\E)}_{\fL'}(\bT\circ\rho) \asim \denote{\E}_\fL(\rho)$ for all
$\E\in \IT_\fL$ and all $\rho:\V\rightarrow\bV$.
\end{definition}
Note that the existence of $\bT$ implies (\ref{related}) in \sect{valid-correct}.

\begin{proposition}{preservation}
Any valid translation up to a preorder $\asim$ preserves $\asim$.
\end{proposition}
\begin{proof}
Let $\fT$ be correct w.r.t.\ the semantic translation $\mathord{\bR} \subseteq \mathord{\asim}$.
Take any $\bT:\bV\rightarrow\bV'$ with $\bT\subseteq \mathord{\bR}$.
Then $\bT\mathord{\circ}\rho \bR \rho$ for any valuation $\rho:\V\rightarrow\bV$.
Hence $\fT$ preserves $\asim$ using the semantic counterpart $\bT$.
\end{proof}
\cor{compositionality} does not extend (from correct and valid translations up to $\asim$) to
translations that preserve $\asim$.

\begin{example}{preserving not compositional}
  Let $\fL$ be a language with unary operators $S$ and $2$,
  and $\fL'$ a language with unary operators $S$ and $G_k$ for each $k>0$.
Their semantics is given by $\bV = \bV' = \IN$, $S^\bV(n)=S^{\bV'}(n)=n+1$ for each $n\in\IN$,
$2^\bV(n)=2n$ for each $n\in\IN$, and $G_k^{\bV'}(n) = 2^kn+1$ for each $n\in\IN$ and $k>0$.

Let $\sim$ be the equivalence relation on $\IN$ defined by $n\sim m$ iff $n$ and $m$ differ
only by a factor $2^k$ for some $k\in\IN$, i.e., $n=m=0$ or their prime factorisations, when leaving
out the factors of $2$, are the same.

Each term $E \in \IT_\fL$ is obtained from a single variable $X$ by applications of the operators $S$
and $2$. Using the equation $2SF = SS2F$, each such term $E$ can be rewritten into a unique \emph{normal form}
$S^i2^jX$ for some $i,j\in\IN$---notation $E \mapsto S^i2^jX$. If $E \mapsto S^i2^jX$ then
$\denote{S^i2^jX}_\fL(\rho) = i+2^j\rho(X) = \denote{E}_\fL(\rho)$ for each valuation $\rho:\V\rightarrow\IN$,
and thus $\denote{S^i2^jX}_\fL \sim \denote{E}_\fL$.

Let $\fT:\IT_\fL \rightarrow \IT_{\fL'}$ be the translation given by
$\left\{\begin{array}{ll}
  \fT(E)=\fT(S^i2^jX) & \mbox{if $E \mapsto S^i2^jX$}\\
  \fT(2^jX) = X \\
  \fT(S^{i+1}X) = S^{i+1}X \\
  \fT(S^{i+1}2^jX) = S^{i}G_{\!j}X.
\end{array}\right.$\\[2ex]
Then $\fT$ preserves $\sim$, taking $\bT$ to be the identity function on $\IN$.
Moreover, the preconditions of \cor{compositionality} are trivially satisfied.
Yet, there is no compositional translation that preserves $\sim$.

For suppose $\fT'$ is such a compositional translation, with semantic counterpart $\bT$.
Then there is a term $E_2$ such that $\fT'(2X)=E_2$. Using that $\fT'$ preserves $\sim$, one has
$\denote{E_2}_{\fL'}(\bT\circ\rho) \sim \denote{2X}_\fL(\rho)$ for each $\rho:\V\rightarrow\IN$.
Take $\rho(X)=0$ and $\rho(Y)=7$ for each $Y\neq X$. Then $\bT(\rho(Y))\sim 7$, so
$\bT(\rho(Y))\geq 7$, for each $Y\neq X$. 
Since $\denote{2X}_\fL(\rho_0)=0$, also $\denote{E_2}_{\fL'}(\bT\circ\rho_0)=0$.
As the operators $S$ and $G_k$ are strictly increasing, it follows that $E_2$ contains none of these
operators, nor variables $Y\neq X$. So $E_2=X$.

As $\fT'$ is compositional, it follows that $\fT'(S2X) =\fT'(SX)$.
Thus, taking $\rho(X)\mathbin=1$, one obtains
$2 = \denote{SX}_{\fL'}(\rho) \sim \denote{\fT'(SX)}_{\fL'}(\bT\circ\rho) =
\denote{\fT'(S2X)}_{\fL'}(\bT\circ\rho) \sim \denote{S2X}_{\fL}(\rho) = 3$, a contradiction.
\end{example}
For this reason, when using preservation of $\asim$ as a
criterion like validity and correctness up to $\asim$, compositionality has to be required separately.

By \obs{compositionality} and \pr{preservation}, assuming that $\fL$ and $\fL'$
satisfy Postulates~\ref{post:head} and~\ref{post:alpha} and $\asim$ respects $\eqa_\fL$ and
$\eqa_{\fL'}$, if a valid fvr-translation from $\fL$ into $\fL'$ up to $\asim$ exists,
then there exists a compositional fvr-translation that preserves $\asim$.
Below I establish the reverse, thereby showing that my concept of a valid translation is quite general.

\begin{theorem}{preservation is valid}
Let $\fL$ and $\fL'$ be languages that satisfy Postulate~\ref{post:alpha},
and $\asim$ a preorder that respects $\eqa_\fL$ and $\eqa_{\fL'}$.
Any compositional fvr-translation $\fT:\IT_\fL\rightarrow\IT_{\fL'}$ that preserves $\asim$
is valid up to $\asim$.
\end{theorem}
\begin{proof}
Let $\fT$ be a compositional fvr-translation that preserves $\asim$, say by means of the semantic
counterpart $\bT\!:\bV\rightarrow\bV'$.
Let $\mathord{\bR}\subseteq\bV'\times\bV$ be the smallest relation containing $\bT^{-1}$ such that $\fT$
is correct w.r.t.\ $\bR$.\linebreak[3] By \lem{correct up to alpha} $\bT$ is correct w.r.t.\ $\bR^\alpha$.
I need to show that $\mathord{\bR^\alpha} \subseteq \mathord{\asim}$.

\emph{Claim:} If $\val'\bR \val$ then
$\val' \eqa_{\fL'} \denote{\fT(\E)}_{\fL'}(\bT\circ\nu)$ and $\denote{\E}_\fL(\nu) \eqa_\fL \val$ 
for some $E\in \IT_\fL$ and $\nu:\V\!\rightarrow\bV\!$.

Establishing the validity of this claim is sufficient, because by \df{preservation} and the
transitivity of $\asim$, using that $\mathord{\eqa_\fL},\mathord{\eqa_{\fL'}}\subseteq \mathord{\asim}$,
it immediately implies that $\mathord{\bR^\alpha} \subseteq \mathord{\asim}$.

I prove the claim with induction on the construction of $\bR$.

\emph{Induction base:} If $\val'\bR \val$ because $\val'=\bT(\val)$ the claim holds by taking $E:=X$
and $\nu(X)=\val$.

\emph{Induction step:} Let $\val' = \denote{\fT(\E)}_{\fL'}(\eta)$ and $\val = \denote{\E}_\fL(\rho)$
for some $\E\in \IT_\fL$, $\eta:\V\rightarrow\bV'$ and $\rho:\V\rightarrow\bV$ with $\eta \bR \rho$.
By induction one may assume that for each $X\in\fv(E)$
there are $E_X\in \IT_\fL$ and $\nu_X:\V\rightarrow\bV$ such that
$\eta(X)\eqa_{\fL'} \denote{\fT(E_X)}_{\fL'}(\bT\circ\nu_X)$ and $\denote{E_X}_\fL(\nu_X) \eqa_\fL \rho(X)$.
Let $(\sigma_X)_{X\in\fv(E)}$ be a family of injective renamings of variables with
$\dom(\sigma_X)=\fv(E_X)$, such that the sets $\range(\sigma_X):=\sigma_X(\fv(E_X))$ for
$X\in\fv(E)$ are pairwise disjoint.
Here I assume the pool of variables to draw from is large enough.
Note that $\denote{\sigma_X}_\fL(\nu_X \circ \sigma_X^{-1}) = \nu_X$
and $\denote{\sigma_X}_{\fL'}(\bT\circ \nu_X \circ \sigma_X^{-1}) = \bT\circ \nu_X$.
Therefore, using (\ref{inductive meaning}),
$\denote{\fT(E_X)}_{\fL'}(\bT\circ\nu_X) \eqa_{\fL'} \denote{\fT(E_X)[\sigma_X]}_{\fL'}(\bT\circ\nu_X\circ\sigma_X^{-1})$ and
$\denote{E_X[\sigma_X]}_\fL(\nu_X\circ\sigma_X^{-1}) \eqa_\fL \denote{E_X}_\fL(\nu_X)$.\vspace{1pt}
By the compositionality of $\fT$, $\fT(E_X)[\sigma_X]\eqa\fT(E_X[\sigma_X])$,
so $\eta(X)\eqa_{\fL'}\denote{\fT(E_X[\sigma_X])}_{\fL'}(\bT\circ\nu_X\circ\sigma_X^{-1})$.

Now let $\nu\!:\!\V\!\rightarrow\bV$ be a valuation such that
$\nu(Y)=\nu_X\circ\sigma_X^{-1}(Y)$ if $Y\in\fv(E_X[\sigma_X])$ for some $X\in\fv(E)$.
Furthermore, define the substitution $\sigma:\fv(E)\rightarrow\IT_\fL$ by $\sigma(X)=E_X[\sigma_X]$.
Then $\rho(X) \eqa_\fL \denote{\sigma(X)}_\fL(\nu)$ for all $X\in\fv(E)$, so
$\val\eqa_\fL \denote{E}_\fL(\denote{\sigma}_\fL(\nu))$ by \lem{alpha congruence}.
Likewise, $\eta(X) \eqa_{\fL'} \denote{\fT(\sigma(X))}_{\fL'}(\bT\circ\nu)$ for all $X\in\fv(E)\supseteq\fv(\fT(E))$,
so $\val'\eqa_{\fL'}\denote{\fT(E)}_{\fL'}(\denote{\fT\circ\sigma}_{\fL'}(\bT\circ\nu))$.
Now (\ref{inductive meaning}) yields $\val\eqa_\fL \denote{E[\sigma]}_\fL(\nu)$ and
$\val'\eqa_{\fL'}\denote{\fT(E)[\fT\circ\sigma]}_{\fL'}(\bT\circ\nu)$.
By compositionality, $\fT(E)[\fT\circ\sigma]\eqa\fT(E[\sigma])$,
from which it follows that $\val'\eqa_{\fL'}\denote{\fT(E[\sigma])}_{\fL'}(\bT\circ\nu)$.
\end{proof}

\begin{corollary}{preservation}
Let $\fL$ and $\fL'$ be languages satisfying Postulates~\ref{post:head} and~\ref{post:alpha},
and $\asim$ a preorder respecting $\eqa_\fL$ and $\eqa_{\fL'}$.
There exists a valid fvr-translation from $\fL$ into $\fL'$ up to $\asim$ iff there exists a
compositional fvr-translation from $\fL$ into $\fL'$ that preserves $\asim$.
\end{corollary}
There exist translations that are compositional and preserve $\asim$
but are not fvr-translations and not valid in the sense of \df{valid}.
Examples are hard to find and do not appear very natural. Moreover, major results I establish
about valid translations (Theorems~\ref{thm:congruence closure of translation} and~\ref{thm:pulling back congruences}) do not
generalise to such examples. For these reasons I prefer to exclude them from my definition of a valid translation.

\begin{example}{preserving but not valid}
Let $\fL$ be a language with constants \textbf{0} and \textbf{1}, and a semantics given by
$\bV\mathbin=\{a,b\}$, \mbox{$\textbf{0}^\bV\mathbin=a$} and $\textbf{1}^\bV=b$.
Let $\fL'$ be a language with a unary operator $f$, and a semantics given by
$\bV'=\{1,2,3,4\}$ and $f^{\bV'}(n)=n{+}1({\rm mod}~4)$.
Finally, let $\sim$ be the equivalence on $\bV\cup\bV'$ given by $a\sim 1 \sim 2 \not\sim 3 \sim 4 \sim b$.

The translation $\fT:\IT_\fL\rightarrow\IT_{\fL'}$ given by 
$\fT(\textbf{0})\mathbin=f(X_0)$ and $\fT(\textbf{1})\mathbin=f(f(f(X_0)))$ for some $X_0\in\V$,
and $\fT(X)\mathbin=X$ for all $X\mathbin\in\V$, is compositional by construction, and preserves $\asim$.
This is witnessed by its semantic counterpart $\bT:\bV\rightarrow\bV'$ given by $\bT(a)=1$ and $\bT(b)=4$.

Since $\T_{\fL'}=\emptyset$, there do not exists fvr-translations from $\fL$ to $\fL'$.
Up to symmetry, $\fT$ is in fact the only translation from $\fL$ to $\fL'$ that preserves $\asim$.
For considering that $f^4(n)=n$ and the choice of $X_0$ is immaterial,
any such translation $\fT_{k,\ell}$ must satisfy $\fT(\textbf{0})\mathbin=f^k(X_0)$ and $\fT(\textbf{1})\mathbin=f^\ell(X_0)$,
where $k,\ell\in\{0,1,2,3\}$.
The only choices for $\bT(a)$ are $1$ or $2$, and by symmetry one can pick $\bT(a)=1$.
Hence, to satisfy $\denote{\fT_{k,\ell}(a)}_{\fL'}(\bT\circ\rho) \asim \denote{a}_\fL(\rho)$ when $\rho(X_0)=a$,
one must take $k\in\{0,1\}$. To also satisfy this formula when $\rho(X_0)=b$, one must take $k\mathbin=1$ and $\bT(b)\mathbin=4$.
This forces $\ell \in \{2,3\} \cap \{3,0\}$, so $\fT_{k,\ell}\mathbin=\fT_{1,3}$.

However, there does not exists a translation from $\fL$ to $\fL'$ that is valid up to $\asim$.
For by symmetry, using \pr{preservation}, $\fT$ is the only candidate for such a translation.
Suppose $\bR$ is a semantic translation w.r.t.\ which $\fT$ is correct. Then $1 \bR a$ or $2 \bR a$.
Suppose $1 \bR a$. Take valuations $\rho,\eta$ with $\eta\bR\rho$, such that $\rho(X_0)=a$ and $\eta(X_0)=1$.
As $\fT$ is correct w.r.t.\ $\bR$, one has
$$2 = f^{\bV'}(1) = \denote{f(X_0)}_{\fL'}(\eta) = \denote{\fT(\textbf{0})}_{\fL'}(\eta) \bR \denote{\textbf{0}}_\fL(\rho) = a.$$
So one must have $2 \bR a$. Now take valuations $\rho,\eta$ with $\eta\bR\rho$, such that $\rho(X_0)=a$ and $\eta(X_0)=2$.
As $\fT$ is correct w.r.t.\ $\bR$, one has
$$3 = f^{\bV'}(2) = \denote{f(X_0)}_{\fL'}(\eta) = \denote{\fT(\textbf{0})}_{\fL'}(\eta) \bR \denote{\textbf{0}}_\fL(\rho) = a.$$
This contradicts the requirement that ${\bR}\subseteq{\sim}$.

\thm{pulling back congruences} does not extend to this translation.
For $\fL$ is a closed-term language and $\fL'$ trivially satisfies \pos{alpha} (for by lack of bound
variables $\eqa_{\fL'}$ is the identity). Any congruence $\approx$ for $\fL'$ contained in $\sim$
must distinguish all four semantic values. Suppose $\fT$ would preserve an equivalence $\approx_\fT$
on $\bV\cup\bV'$, contained in $\approx$, that on $\bV'$ coincides with $\approx$.
Let $\bT'$ be a semantic counterpart of $\fT$.
Then $\bT'(a)=1$ or $\bT'(a)=2$.
Suppose $\bT'(a)=1$. Take a valuation $\rho$ with $\rho(X_0)=a$.
Then
$$2 = f^{\bV'}(1) = \denote{f(X_0)}_{\fL'}(\bT'\circ\rho) = \denote{\fT(\textbf{0})}_{\fL'}(\bT'\circ\rho) \approx_\fT \denote{\textbf{0}}_\fL(\rho) = a.$$
So one has $2 \approx_\fT a \approx_\fT 1$, a contradiction. The case $\bT'(a)=2$ leads to a contradiction
in the same way.
Hence $\fT$ does not preserve such an equivalence $\approx_\fT$.
\end{example}

\section{Validity up to barbed bisimilarity for process calculi}\label{sec:barbed}

In this paper I defined a notion of a valid translation op to a semantic equivalence $\sim$
between general system description languages, dealing in \sect{closed-term} with the special case of
closed-term languages.
In this section I zoom in further on process calculi such as CCS \cite{Mi90ccs} and the
$\pi$-calculus \cite{SW01book},
and ask which equivalences $\sim$ to use in studying their relative expressiveness.
These languages have in common that their semantics can be given in terms of labelled transition
systems, whose states, called \emph{processes}, are the closed expressions in the language, and
whose transitions are given by an operational semantics in the style of Plotkin~\cite{Pl04}.

For a wide class of process calculi without name-binding, the finest equivalence in regular employ
is strong bisimilarity \cite{Mi90ccs}. In proving that process calculus $\fL'$ is at least as
expressive as process calculus $\fL$ up to a semantic equivalence $\sim$, the choice of 
strong bisimilarity for $\sim$ thus yields the strongest result.
Accordingly, Robert de Simone \cite{dS85} showed that a wide class of process calculi, including
CCS \cite{Mi90ccs}, CSP \cite{BHR84}, ACP \cite{BK86a} and SCCS~\cite{Mi83}, are expressible up to
strong bisimilarity in {\sc Meije}~\cite{AB84}.

However, strong bisimilarity is not suitable for comparing any two process calculi.
As an example, consider the $\pi$-calculus with the early operational semantics versus the 
$\pi$-calculus with the late operational semantics \cite{SW01book}.
Both operational semantics are meant to convey the same idea on how $\pi$-calculus processes
interact. Consequently, one would hope and expect that the identity function between these versions
of the $\pi$-calculus is a valid translation. However, it is not valid up to strong bisimilarity.
The late semantics has transitions labelled $x(y)$, where the name $y$ is a variable in which a
value received on channel $x$ will be stored, whereas the early semantics has transitions labelled
$xy$, where the name $y$ is a particular value received on channel $x$. Since these labels have a
different shape, and strong bisimilarity requires matching labels, validity up to strong
bisimilarity fails. Moreover, the problem cannot be resolved by simply renaming the labels, for the
entire meaning of input transitions is different.

In the $\pi$-calculus, \emph{strong early bisimilarity} \cite{SW01book} is a much more canonical
semantic equivalence than strong bisimilarity, and the identity function \emph{is} a valid
transition between the early and late $\pi$-calculi up to strong early bisimilarity. However, the
definition of early bisimilarity for the $\pi$-calculus with the late operational semantics
\cite{SW01book} has a very different form than the definition of early bisimilarity for the early
$\pi$-calculus \cite{SW01book}.  The reason is that the same idea needs to be phrased in terms of
rather different transition relations.  Consequently, taking strong early bisimilarity as a unifying
equivalence for comparing the late and early $\pi$-calculus appears to be somewhat ad hoc.
Moreover, neither strong bisimilarity nor strong early bisimilarity would be suitable to compare,
say, the $\pi$-calculus with CCS\@.

A canonical semantic equivalence that can be defined in a uniform way on CCS as well as on the late
and early $\pi$-calculi is \emph{strong barbed bisimilarity}, originally proposed by Milner and
Sangiorgi in~\cite{MilS92}.
It is based on the idea that a $\tau$-transition of a process $P$ describes an actual reaction
of the process $P$, whereas a translation labelled $a \neq \tau$ merely indicates a
potential reaction of $P$ when synchronising with a communication partner
that is willing to engage in the complementary transition $\bar a$. (Likewise, in CSP \cite{BHR84} a
transition $P\stackrel{a}\longrightarrow P'$ can be regarded as a potential; one that is not
realised when putting the process in a parallel composition involving synchronisation on the action
$a$ when the other component cannot partake in such a synchronisation. The only way to be sure
that the CSP action $a$ cannot be inhibited by any synchronisation context is to \emph{conceal} it
from the environment, using the CSP concealment operator $\_\!\_\, \backslash a$, renaming the
action $a$ into $\tau$.) From this perspective, it appears natural to formulate semantic
equivalences entirely in terms of the $\tau$-transitions processes can do, and capture the
communication potentials (manifested by the other transitions) merely by studying the behaviour of
processes in contexts.
Accordingly, \emph{(strong) reduction bisimilarity} was defined in \cite{MilS92} as the version of
strong bisimilarity that requires the transfer property for $\tau$-transitions only, i.e., strong
barbed bisimilarity as defined below, without the first clause, and reading
\plat{$\stackrel\tau\longrightarrow$} for $\rightarrow$.
Naturally, reduction bisimilarity is not a congruence: the CCS processes $a.\textbf{0}$ and
$\textbf{0}$ are reduction equivalent (for neither can do a $\tau$-transition), yet
$a.\textbf{0} | \bar a.\textbf{0}$ and $\textbf{0} | \bar a.\textbf{0}$ are not.
The purpose of the relation is to define a reasonable semantic equivalence as its congruence
closure. Indeed, for divergence-free CCS processes strong bisimilarity turned out to be the
congruence closure of reduction bisimilarity \cite{MilS92}.
Unfortunately, this failed for processes with divergences (infinite-sequences of
$\tau$-transitions) \cite{MilS92}, and when lifted to the weak case, weak reduction congruence turned
out to be the universal relation, and thus useless \cite{MilS92}.
For this reason, reduction bisimilarity needed to be strengthened.
The main idea is that processes like $a.\textbf{0}$ and $\textbf{0}$ can be told apart by placing
them in a context $\_\!\_ | \bar a.\omega$, where $\omega$ denotes a success state that might be reached
by the environment in which a process is placed. The process $a.\textbf{0} | \bar a.\omega$
can reach this success state by performing a $\tau$-transition, whereas $\textbf{0} | \bar a.\omega$ can not.
Writing $\sbarb{P}{\omega}$ to indicate that a process $P$ has reached this success state,
reduction bisimilarity can be strengthened by requiring that if $\sbarb{P}{\omega}$
then any process $Q$ equivalent to $P$ should also satisfy $\sbarb{Q}{\omega}$.
This yields the notion of strong barbed bisimilarity. The predicate $\sbarb{\_\!\_\,}{\omega}$ is
called a \emph{barb}. In general, one can use a collection of barbs.

Barbed bisimilarity is defined on closed-term languages $\fL$ that are equipped with a binary
\emph{reduction relation} ${\rightarrow}\subseteq\T_\fL\times\T_\fL$ between processes
and with a collection $\{ \sbarb{}{\omega} \subseteq \T_\fL \mid \omega \in\Omega\}$ of unary predicates on
processes, called \emph{barbs}.

\begin{definition}{sbb}
\emph{Strong barbed bisimilarity} is the largest symmetric relation ${\sbb}\subseteq \T_\fL \times \T_\fL$
such that
\begin{itemize}
\item $P \sbb Q$ and $\sbarb{P}{\omega}$ with $b\in\Omega$ implies $\sbarb{Q}{\omega}$, and
\item $P \sbb Q$ and $P \rightarrow P'$ implies $Q \rightarrow Q'$ for some $Q'$ with $P' \sbb Q'$.
\end{itemize}
\emph{Strong barbed congruence}, $\sbb^{1c}_\fL$, is the congruence closure of $\sbb$ w.r.t.\ the
language $\fL$.
\end{definition}
For process calculi equipped with a labelled transition system semantics, the reductions
$P \rightarrow P'$ are simply the $\tau$-transitions. However, for many languages it is possible to
give a \emph{reduction semantics} that generates the reductions directly, without first constricting
a labelled transition system. Some languages, such as the $\lambda$-calculus, come with a natural
reduction semantics, even when they have no labelled transition semantics at all.

The definition of the barbs on various process calculi is somewhat ad hoc.
In \cite{MilS92} only one barb was used, and a process has this barb when it can perform an
observable action, i.e., $P\sbarb{}{}$ iff \plat{$P \stackrel{a}\longrightarrow P'$} for some process $P'$
and action $a\neq \tau$. An alternative is to introduce a barb $a$ for each action $a \neq\tau$,
taking $P\sbarb{}{a}$ iff $P \stackrel{a}\longrightarrow P'$ for some $P'$.
Both forms of barbs lead to the same notion of strong barbed congruence.\vspace{2pt}
In the $\pi$-calculus, $x$ and $\bar x$ for $x\in\N$ are taken to be barbs (or just $\bar x$), with 
$P\sbarb{}{\omega}$ iff \plat{$P \stackrel{\omega(y)}\longrightarrow P'$} or
 \plat{$P \stackrel{\omega y}\longrightarrow P'$} for some $P'$---cf.\ \sect{asynpi}.
Again, including barbs $x\in\N\!$ or not makes no difference in the resulting notion of strong barbed congruence.
In fact, it does not matter at all how the barbs are defined, but only whether there are enough barbs, 
and how they propagate upwards through a context (that is, how the validity of $\sbarb{C[P]}{\omega}$ is
determined by the validity of $\sbarb{P}{\omega}$). This is because in determining whether two processes
$P$ and $Q$ are strongly barbed congruent, the barbs that help in this determination appear in
contexts in which the processes $P$ and $Q$ are placed, rather than in $P$ and $Q$ themselves.
For this reason, aiming for a fairly language-independent definition of barbed congruence, I propose
the use of \emph{external barbs} instead, defined as follows.
Postulate a sufficiently large collection $\Omega$ of barbs and add each of them as a fresh
constant to the language under consideration.  Require them to propagate upwards through a context
in the same way as $\tau$-transitions or reduction steps. Thus, whenever the structural operational
semantics of the language has a rule\vspace{-1.5ex}
$$\frac{P_i \rightarrow Q}{f(P_1,\dots,P_n) \rightarrow R}$$
for some $n$-ary operator $f$, possibly with \plat{$\stackrel\tau\longrightarrow$} in the r\^ole of
$\rightarrow$, then postulate a rule
$$\frac{\sbarb{P_i}{\omega}}{\sbarb{f(P_1,\dots,P_n)}{\omega}}$$
for each barb $\omega\in\Omega$. Note that this agrees exactly with \df{barbs} of (strong) output
barbs in the $\pi$-calculus.
It turns out that for CCS, the $\pi$-calculus, and in fact any language with a definition of
strong barbed congruence that I am aware of, the above use of external barbs yields the same notion
of strong barbed congruence as the original approach. It may be felt as a drawback that in defining
strong barbed congruence on a language $\fL$, $\fL$ needs to be extended, namely by the added
constants $\omega\in\Omega$. However, as a tool to define strong barbed congruence on $\fL$ this is
not a big problem, as the definition on the extended language naturally restricts to $\fL$.

The treatment of barbs proposed above is strongly inspired by the success action $\omega$ in the
treatment of \emph{testing equivalences} by De Nicola \& Hennessy \cite{DH84}, and by the criterion
of \emph{success sensitiveness} imposed by Gorla \cite{Gorla10a} on translations between process calculi.
Both \cite{DH84} and \cite{Gorla10a} use only a single success predicate (barb).
In \cite{Gorla10a} $\sbarb{P}{\omega}$ is defined to hold iff $P$ has a ``top-level unguarded occurrence''
of $\omega$. Gorla defines the latter concept only for languages that are equipped with a notion of
\emph{structural congruence} $\equiv$ as well as a parallel composition $|$. In that case
$P$ has a top-level unguarded occurrence of $\omega$ iff $P\equiv Q|\omega$, for some $Q$~\cite{Gorla10a}.
  Specialised to the $\pi$-calculus, a \emph{(top-level) unguarded} occurrence is one that not lays
  strictly within a subterm $\alpha.Q$, where $\alpha$ is $\tau$, $\bar xy$ or $x(z)$.
As far as I know, for all languages where Gorla's definition of $\sbarb{P}{\omega}$ as well as mine
apply, they yield the same result.

On CCS strong barbed congruence coincides with strong bisimilarity \cite{MilS92}.
On the $\pi$-calculus strong barbed congruence coincides with strong early congruence, the
congruence closure of strong early bisimilarity \cite{SW01book}.
This testifies to the success of the barbed bisimilarity approach to defining canonical semantic
equivalences for different process calculi in a uniform way. For many other calculi, strong barbed
congruence is the primary \emph{proposal} for a canonical semantic equivalence, and its explicit
characterisation in a manner that does not involve quantification over all context is a task that
merely follows.

This makes strong barbed bisimilarity an attractive candidate for the equivalence relation $\sim$ up
to which translations between process calculi are proven valid. \thm{congruence closure of translation}
says that if a translation between two process calculi is valid up to strong barbed bisimilarity,
then it is also valid up an equivalence that on the source language coincides with strong barbed
congruence, and on the image of the source language within the target language is also the
congruence closure of strong barbed bisimilarity.

For ``nonprompt encodings'' \cite{NestmannP00}, that ``allow administrative (or book-keeping) steps
to precede a committing step'', strong (barbed) bisimilarity is too fine an equivalence to relate
source processes and their translations. Here an appealing alternative is
\emph{weak barbed bisimilarity}, cf.\ \df{barbed congruence}.
On CCS weak barbed congruence coincides with weak bisimilarity as defined in \cite{Mi90ccs}.
On the $\pi$-calculus weak barbed congruence coincides with weak early congruence as used in
\cite{SW01book}, at least for image finite processes, or in general when allowing infinite sums in
the $\pi$-calculus \cite{SW01}. \sect{asynpi} mentioned Boudol's encoding of $\pi$ into $a\pi$ as
an example of a transition valid up to weak barbed bisimilarity.

In many situations one would reject translations between process calculi, that introduce (or
eliminate) divergences. This can be captured in my framework by requiring validity op to
\emph{weak (barbed) bisimilarity with explicit divergence} \cite{BKO87}.
It adds to \df{barbed congruence} the clause
\begin{itemize}
\item $P \wbb Q$ and $P{\Uparrow}$ implies $Q{\Uparrow}$.
\end{itemize}
Here $P{\Uparrow}$ denotes that there are $P_i$ for $i\geq 0$ such that $P=P_0$ and
$P_i\rightarrow P_{i+1}$ (or $P_i\stackrel\tau\longrightarrow P_{i+1}$) for all $i\geq 0$.
Using Remark 1, Theorem 2, Observation 3 and Corollary 3 in \cite{vG18a}, the proof of
Theorem 2 in \cite{LMGG18} shows that Boudol's encoding of $\pi$ into $a\pi$ is even valid up to
weak barbed bisimilarity with explicit divergence.

In most industrial applications of process calculi the r\^ole once played by weak bisimilarity has
largely been taken over by branching bisimilarity \cite{GW96}---cf.\ \cite{GLMS11,CranenEtAl13}---of
which both a default version and two versions with explicit divergence are in use \cite{GLT09b,FGL17a}.
The following proposal captures these through the barbed bisimilarity methodology:

\begin{definition}{bbb}
\emph{Divergence-preserving branching barbed bisimilarity} is the largest symmetric relation ${\bbbis}\subseteq \T_\fL \times \T_\fL$
such that
\begin{itemize}
\item $P \bbbis Q$ and $\sbarb{P}{\omega}$ with $\omega\in\Omega$ implies $Q \Longrightarrow Q^\dagger$ for
  some $Q^\dagger$ with $P \bbbis Q^\dagger$ and $\sbarb{Q^\dagger}{\omega}$,
\item $P \bbbis Q$ and $P \rightarrow P'$ implies $Q \Longrightarrow Q^\dagger \mathrel{({\rightarrow})} Q'$
  for some $Q^\dagger,Q'$ with $P \bbbis Q^\dagger$ and $P' \bbbis Q'$, and
\item $P \bbbis Q$ and $P \rightarrow P_1 \rightarrow P_2 \rightarrow \dots$ implies
  $Q \rightarrow Q'$ for some $Q'$ with $P_k \bbbis Q'$ for some $k\geq 0$.
\end{itemize}
\end{definition}
Here $P \mathrel{({\rightarrow})} Q$ abbreviates $(P=Q) \vee (P \rightarrow Q)$.
See \cite{GLT09b} for a number of equivalent versions of the last clause,
i.e.\ with $Q \Longrightarrow\rightarrow Q'$.
For \emph{weakly divergence-preserving branching barbed bisimilarity} the last clause is weakened to
``$P \bbbis Q$ and $P{\Uparrow}$ implies $Q{\Uparrow}$'', just as for weak barbed bisimilarity with
explicit divergence above. Skipping the last clause altogether yields \emph{branching barbed bisimilarity}.
On CCS, the congruence closures of these equivalences yield
\emph{((weakly) divergence-preserving) branching bisimilarity}, as defined in the literature \cite{FGL17a};
the proofs are not essentially different from the ones for strong and weak barbed bisimilarity.
Using the remark below Observation 3 in \cite{vG18a}, the proof of Theorem 2 in \cite{LMGG18}
can be adapted to show that Boudol's encoding of $\pi$ into $a\pi$ is even valid up to
divergence-preserving branching barbed bisimilarity.

\section{Related work: full abstraction}\label{sec:full abstraction}

The concept of \emph{full abstraction} stems from Milner \cite{Mi75}.
It indicates a particularly nice connection between a denotational and an operational semantics of a
language $\fL$. Here a denotational semantics is a function $\denote{\ \ }_\fL$ as introduced in
\sect{validity}, whereas an operational semantics is given by an \emph{evaluation function}
$\mathcal{E}:\T_\fL \rightarrow\IO$ from the closed terms, there called \emph{programs}, to a set
$\IO$ of \emph{observations}. Evaluation determines an equivalence relation on programs:
$\p \sim_{\mathcal{E}} Q$ iff $\mathcal{E}(\p) = \mathcal{E}(\p)$. Let $\sim_{\mathcal{E}}^{1c}$
be the congruence closure of $\sim_{\mathcal{E}}$ for the language $\fL$, as defined in
\sect{congruence closure}.

\begin{definitionc}{full abstraction}{\cite{Mi75}}
The semantic function $\denote{\ \ }_\fL$ for $\fL$ is \emph{fully abstract} w.r.t.\ $\mathcal{E}$
iff for all $\p,Q\in\T_\fL$
\[\denote{\p}_\fL = \denote{Q}_\fL ~~~\Leftrightarrow~~~ \p \sim_{\mathcal{E}}^{1c} Q\,.\]
\end{definitionc}
A semantic function $\denote{\ \ }_\fL$ always induces an equivalence relation on $\T_\fL$ by
$\p \approx_\fL Q$ iff $\denote{\p}_\fL = \denote{Q}_\fL$. When this semantics is defined inductively
as indicated in \sect{closed-term}, $\approx_\fL$ must be a congruence. Now $\denote{\ \ }_\fL$ is
fully abstract w.r.t.\ $\mathcal{E}$ iff ${\approx_\fL} = {\sim_{\mathcal{E}}^{1c}}$.

Note that any equivalence relation $\sim_{\mathcal{E}}$ can be extracted from an evaluation
function $\mathcal{E}:\T_\fL \rightarrow\IO$, namely by taking $\IO$ to be the set of 
$\sim_{\mathcal{E}}$-equivalence classes of closed terms, with $\mathcal{E}$ mapping each
$p\in\T_\fL$ to its own equivalence class. Likewise, each congruence relation $\approx$ on $\T_\fL$
can be obtained from a semantics $\denote{\ \ }_\fL$: take $\bV$ to be the set of
$\approx$-equivalence classes of closed terms, and for $E\in\IT_\fL$ and $\rho:\fL\rightarrow\bV$ define
$\denote{E}_\fL(\rho)$ by to be the $\approx$-equivalence class of $E[\sigma]$,
where $\sigma:\V\rightarrow\T_\fL$ is a closed substitution that maps each variable $X$ to a member of the
$\approx$-equivalence class of closed terms $\rho(X)$. Since $\approx$ is a congruence, this
definition is independent of the choice of $\sigma$.

Consequently, full abstraction can equally well be stated as a relation between two equivalence
relations $\approx$ and $\sim$ on $\T_\fL$:\vspace{-2ex}
\begin{center}
$\approx$ is \emph{fully abstract} w.r.t.\ $\fL$ and $\sim$ ~~~iff~~~ ${\approx} = {\sim_{\fL}^{1c}}$.
\end{center}
It is in this spirit that full abstraction has been employed in \cite{vG93d} and subsequent papers.

\newcommand{\So}{{\rm S}}
\newcommand{\Ta}{{\rm T}}
Riecke \cite{Rie91} and Shapiro \cite{Sha91} extend the notion of full abstraction to translations
between languages.
Riecke \cite{Rie91} compares languages $\fL_\So$ and $\fL_\Ta$ with a shared evaluation function
$\mathcal{E}:\T_{\fL_\So} \cup \T_{\fL_\Ta} \rightarrow 2^{\IOs}$, associating with each closed expression
a set of observations.
Write $\p \sqsubseteq^\mathcal{E} Q$, for $\p,Q\in\T_{\fL_\So} \cup \T_{\fL_\Ta}$, if
$\mathcal{E}(\p) \subseteq \mathcal{E}(Q)$, and let $\sqsubseteq^\mathcal{E}_{\fL_i}$ be the
congruence closure of $\sqsubseteq^\mathcal{E}$ w.r.t.\ $\fL_i$, for $i=\So,\Ta$ (source and target).
He calls a translation $\fT : \IT_{\fL_{\rm S}} \rightarrow \IT_{\fL_{\rm T}}$ fully abstract iff,
for all $\p, Q \in \T_{\fL_\So}$,
\[ \p \sqsubseteq^\mathcal{E}_{\fL_\So} Q ~~~\Leftrightarrow~~~ \fT(\p) \sqsubseteq^\mathcal{E}_{\fL_\Ta} \fT(Q) \,.\]
The same notion occurs earlier in Mitchell \cite{Mitchell93}, although not under the name ``full abstraction'',
and using equivalence relations instead of preorders---taking $\p \sim^\mathcal{E} Q$ iff
$\mathcal{E}(\p) = \mathcal{E}(Q)$. He compares the expressive power of programming languages in
terms of \emph{abstraction preserving reductions} between them. These are translations that are
compositional as well as fully abstract in the above sense.
Later work abstracts from the evaluation function $\mathcal{E}$, and casts full abstraction directly
in terms of equivalences $\sim^\So$ and $\sim^\Ta$ on the source and target languages \cite{San93}.

Felleisen \cite{Fel91} compares the expressive power of programming languages through a notion of
\emph{eliminability} of language constructs. With some effort, this approach can be seen to have
significant similarities with the approach of Mitchell \cite{Mitchell93}, although it allows certain
degenerate reductions \cite{Mitchell93}.

Whereas most work on expressiveness deals with closed-term languages, allowing a focus on syntax
over semantics, Shapiro \cite{Sha91} works entirely on the semantic side, leaving the syntax largely implicit.
His languages are triples $(\bV,\fL,\sim)$, consisting of a semantic domain $\bV$, a collection
$\fL$ of partial operators on $\bV$, and a semantic equivalence $\sim$ on $\bV$.
A \emph{language embedding} of one such language $(\bV_\So,\fL_\So,\sim^\So)$ into another $(\bV_\Ta,\fL_\Ta,\sim^\Ta)$
is defined to be a homomorphism of the partial algebra $(\bV_\So,\fL_\So)$ into the partial algebra $(\bV_\Ta,\fL_\Ta)$.
Recasting this definition in terms of my framework, this is a function $\bR:\bV_\So\rightarrow\bV_\Ta$
such that there exists a compositional translation $\fT:\fL_1\rightarrow\fL_2$ correct w.r.t.\ $\bR$
(as in \df{correct R}). For $i=\So,\Ta$ let $\sim^i_{\fL_i}$ be the congruence closure of
$\sim^i$ w.r.t.\ $\fL_i$---Shapiro calls this the \emph{fully-abstract} congruence of $\fL_i$. Then
a language embedding is deemed fully abstract iff, for all $\val,\wal\in\bV_\So$,
$$ \val \sim^\So_{\fL_\So} \wal ~~~\Leftrightarrow~~~ \fT(\val) \sim^\Ta_{\fL_\Ta} \fT(\wal) \,.$$
In \cite{Sha92} these fully abstract language embeddings are used to classify a number of concurrent
programming languages on expressive power.

In Nestmann \& Pierce \cite{NestmannP00} the notion of full abstraction was generalised by dropping the
requirement that the source and target equivalences being compared need to be congruence closures.

\begin{definition}{Full Abstraction}
A translation $ \fT \!: \IT_{\fL_\So} \rightarrow \IT_{\fL_\Ta} $ is \emph{fully abstract} w.r.t.\ the equivalences
$ {\sim_\So} \subseteq \T_{\fL_\So}^2 $ and $ {\sim_\Ta} \subseteq \T_{\fL_\Ta}^2 $ if, for all
$ \p, Q \in \T_{\fL_\So} $,\vspace{-2ex}
$$ \p \sim_\So Q ~~~\Leftrightarrow~~~ \fT(\p) \sim_\Ta \fT(Q) \,.$$
\end{definition}
In this form, full abstraction has found widespread applications \cite{Nestmann00,BPV05}.
Fu~\cite{Fu16}, for instance, bases a theory of expressiveness on full abstraction, with
divergence-preserving branching barbed bisimilarity in the r\^ole of $\sim_\So$ and $\sim_\Ta$.

As stressed in \cite{GN16,Parrow16}, the notion of full abstraction is meaningful only in relation
to a well-chosen pair of source and target equivalences, and only in combination with a criterion
like compositionality.
In particular, for each encoding $\fT$ and each target equivalence $ \sim_\Ta $ there
exits a source term equivalence $ \sim_\So $, namely $\{(\p, Q) \mid \fT(\p) \sim_\Ta \fT(Q)\}$,
such that $\fT$ is fully abstract w.r.t.\ $\sim_\So$ and $\sim_\Ta $.
For each injective encoding $\fT$ and each source term relation $ \sim_\So$, there exits
$ {\sim_\Ta} \subseteq \T_{\fL_\Ta}^2 $, namely $\{ (\fT(\p),\fT(Q)) \mid \p\sim_\So Q\}$,
such that $\fT$ is fully abstract w.r.t.\ $\sim_\So$ and $\sim_\Ta $.
Finally, for each pair $ \sim_\So $ and $\sim_\Ta $ such that the cardinality of
$\T_{\fL_\Ta}/_{\sim_\Ta}$ greater than or equal to the cardinality of $\T_{\fL_\So}/_{\sim_\So}$
there exists a translation from $\fL_\So$ to $\fL_\Ta$ that is fully abstract w.r.t.\ $\sim_\So$ and $\sim_\Ta $.

Naturally, any translation that is valid up to an equivalence $\sim$ as in \df{valid} of the current
paper is also fully abstract, namely w.r.t.\ the same equivalence $\sim$ in the r\^ole of both
$ \sim_\So $ and $\sim_\Ta $. Furthermore, validity entails compositionality through \cor{compositionality}.
Both full abstraction and validity imply an injective translation from the 
$\sim_\So$-equivalence classes of closed source terms to the $\sim_\Ta$-equivalence classes of
closed target terms. However, validity demands that this link between source and target terms is
again governed by $\sim$, whereas full abstraction implies no counterpart to this crucial requirement.

A typical example of a valid transition is the encoding of the synchronous into the asynchronous
$\pi$-calculus described in \sect{asynpi}. This encoding is valid up to weak barbed bisimilarity $\wbb$.
Consequently is it also fully abstract w.r.t.\ $\wbb$ and $\wbb$ according to \df{Full Abstraction}.
However, it is not fully abstract w.r.t.\ $\wbb$ and $\wbb$ in the more original sense of
Shapiro~\cite{Sha91} and others, due to the fact that $\wbb$ is not a congruence for either the
source or the target language. By \thm{congruence closure of translation} the same translation is
also fully abstract w.r.t\ the congruence closure of $\wbb$ on the source language, and a somewhat artificial
equivalence on the target language that is the congruence closure of $\wbb$ under translated source contexts.
Although closer, this is still not a full-abstraction result according to Shapiro, as the latter
equivalence fails to be a congruence for all of the target language. Finally,
by \thm{pulling back congruences}, the same encoding is furthermore fully abstract w.r.t.\ the
congruence closure of $\wbb$ on the target language, and an artificial congruence on the source
language that is strictly finer than the congruence closure of $\wbb$ on the source language.
This is a full abstraction result in the framework of \cite{Sha91}. However, in line with the
observations of \cite{GN16,Parrow16}, the latter two full abstraction results should not be regarded
as lending additional credibility to this particular encoding of the synchronous into the asynchronous
$\pi$-calculus.
Since Theorems~\ref{thm:congruence closure of translation} and~\ref{thm:pulling back congruences}
hold in great generality, these full abstraction results are merely consequences of the validity of
the encoding up to $\wbb$.

It is well known that Boudol's encoding---described in \sect{asynpi}---of the synchronous into the
asynchronous $\pi$-calculus fails to be fully abstract w.r.t.\ $\cong^c$ and $\cong^c_a$.
Here $\cong^c$ is the congruence closure of $\wbb$ on the source language, and
$\cong^c_a$ the congruence closure of $\wbb$ on the target language.
A counterexample was given at the end of \sect{ccproperty}.
The same analysis applies to a similar encoding proposed by Honda \& Tokoro \cite{HT91}.
In \cite{DYZ18} this problem is addressed by proposing a strict subcalculus $SA\pi$ of the 
target language that contains the image of the source language under of a version Honda
\& Tokoro's encoding, such that this encoding is fully abstract w.r.t.\ $\cong^c$ and
the congruence closure of $\wbb$ w.r.t.\ $SA\pi$.
In \cite{QW00} a similar solution to the same problem was found earlier, but for a variant of
Boudol's encoding from the \emph{polyadic} $\pi$-calculus to the (monadic) asynchronous
$\pi$-calculus.  They defined a class of \emph{well-typed} expressions in the asynchronous
$\pi$-calculus, such that the well-typed expressions constitute a subcalculus of the target language
that contains the image of the source language under the encoding. Again, the encoding is fully
abstract w.r.t.\ $\cong^c$ and the congruence closure of $\wbb$ w.r.t.\ that sublanguage.
By \thm{congruence closure of translation} such results can always be achieved, namely by taking
as target language exactly the image of the source language under the encoding.
What the results of \cite{QW00,DYZ18} add is that the sublanguage may be strictly larger than the
image of the source language, and that its definition is not phrased in terms of the encoding.

\section{Related work: validity of encodings according to Gorla}\label{sec:related}

In the last twenty years, a great number of encodability and separation results have
appeared, comparing CCS, Mobile Ambients, and several versions of the $\pi$-calculus (with
and without recursion; with mixed choice, separated choice or asynchronous)
\cite{
San93,           
Jen94,           
Boreale98,        
Parrow00,       
NestmannP00,      
Palamidessi03,  
Nestmann00,       
CardelliG00,    
CardelliGG02,   
BusiGZ09,       
CarboneM03,       
BPV04,            
BPV05,            
Palamidessi05,  
PalamidessiSVV06, 
PV06,           
CCP07,          
VPP07,          
CCAV08,         
HMP08,          
PV08,           
VBG09,          
EPTCS160.2,      
EPTCS190.5,      
PN16,            
PP16
};     
see \cite{Gorla10b,Gorla10a} for an overview. Many of these results employ different and
somewhat ad hoc criteria on what constitutes a valid encoding, and thus are hard to
compare with each other. Several of these criteria are discussed and compared in \cite{Nestmann06},
\cite{Parrow08} and \cite{Peters12}. Gorla \cite{Gorla10a} collected some essential features of
these approaches and integrated them in a proposal for a valid encoding that justifies
most encodings and some separation results from the literature.
Since then, several authors have used Gorla's framework as a basis for
establishing new valid encodings and separation results
\cite{Gorla10b,LPSS10,PSN11,PN12,PNG13,GW14,EPTCS160.4,EPTCS189.9,GWL16}.

Like Boudol \cite{Bo85} and the present paper, Gorla requires a compositionality condition
for encodings. However, his criterion differs on two counts from mine.
It is stronger in that the expression $E_f$ encoding an operator $f$ may use each variable
corresponding to an argument of $f$ only once. A translation defined (in part) by
$\fT(f(P)) = g(\fT(P)\|\fT(P))$ for instance can be compositional according to
\df{compositionality}, but not according to Gorla's definition. It is weaker than mine in
that the expression $E_f$ encoding an operator $f$ may be dependent on the set of names
occurring freely in the expressions given as arguments of $f$. This issue is further discussed in \cite{vG12}.
It is an interesting topic for future research to see if there are any valid encodability
results \`a la \cite{Gorla10a} that suffer from my proposed strengthening of compositionality.

The second criterion of \cite{Gorla10a} is a form of invariance under name-substitution.
It serves to partially undo the effect of making the compositionality requirement
name-dependent. In my setting I have not yet found the need for such a condition.
In \cite{vG12} I argue that this criterion as formalised in \cite{Gorla10a} is too restrictive.

The remaining three requirements of Gorla (the `semantic' requirements) are very close to an
instantiation of mine with a particular preorder $\asim$.
If one takes $\asim$ to be weak barbed bisimilarity with explicit divergence
(i.e.\ relating divergent states with divergent states only---see \sect{barbed}), using external barbs
as defined in \sect{barbed}, then any valid translation in my sense satisfies Gorla's semantic
criteria, provided that the equivalence $\asymp$ on the target language that acts as a parameter in
Gorla's third criterion is also taken to be weak barbed bisimilarity with explicit divergence.
The precise relationships between the proposals of \cite{vG12} and \cite{Gorla10a} are further
discussed in \cite{EPTCS190.4}.

Further work is needed to sort out to what extent the two approaches have
relevant differences when evaluating encoding and separation results from the literature.
Another topic for future work is to sort out how dependent known encoding and separation
results are on the chosen equivalence or preorder. 

\bibliographystyle{eptcs}
\bibliography{$HOME/Stanford/lib/abbreviations,$HOME/Stanford/lib/dbase,glabbeek,$HOME/Stanford/lib/new,pi}

\end{document}